\documentclass[reprint,
superscriptaddress,
nofootinbib,
 amsmath,amssymb,
 aps,
]{revtex4-2}

\usepackage[english]{babel}
\usepackage[utf8]{inputenc}
\usepackage[T1]{fontenc}
\usepackage{times}
\usepackage{graphicx}
\usepackage{dcolumn}
\usepackage{comment}
\usepackage{bm}
\usepackage{booktabs}
\usepackage{color}
\usepackage{mathtools}
\usepackage{physics}
\usepackage{cancel}
\usepackage{braket}

\usepackage{amsthm}

\newtheorem{lemma}{Lemma}

\usepackage{tabularray}

\usepackage{array}
\newcolumntype{M}[1]{>{\centering\arraybackslash}m{#1}}
\newcolumntype{P}[1]{>{\centering\arraybackslash}p{#1}}

\usepackage{ragged2e} 
\usepackage{caption}
\usepackage[position=top]{subcaption}
\DeclareCaptionJustification{justified}{\justifying}
\usepackage[normalem]{ulem}
\usepackage[export]{adjustbox}

\usepackage{pgfplots}
\usepgfplotslibrary{groupplots,dateplot}
\usetikzlibrary{patterns,shapes.arrows}
\pgfplotsset{
    compat = newest,
    legend style = {font = \scriptsize},
    label style = {font = \footnotesize},
    tick label style = {font = \scriptsize},
    width = 0.99\columnwidth,
    /tikz/mark size=1.5pt,
}

\usetikzlibrary{external}
\usepackage{hyperref}
\hypersetup{
    colorlinks,
    linkcolor={violet},
    citecolor={teal},
    urlcolor={teal}
}

\newcommand{\twoOrd}[1]{\, \mathrel{:} \mathrel{#1} \mathrel{:} \,}
\newcommand{\threeOrd}[1]{\mathrel{\lower.04em\hbox{\rlap{$\cdot$}}{:}} \, \mathrel{#1} \, \mathrel{\lower.04em\hbox{\rlap{$\cdot$}}{:}}}

\newcommand{\FUB}{Dahlem Center for Complex Quantum Systems and Institut für Theoretische Physik, Freie Universität Berlin, Arnimallee 14, 14195 Berlin, Germany}

\newcommand{\HZB}{Helmholtz-Zentrum Berlin für Materialien und Energie, Hahn-Meitner-Platz 1, 14109 Berlin, Germany}

\newcommand{\UOC}{Institute of Theoretical and Computational Physics, Department of Physics, University of Crete, 71003 Heraklion, Greece}

\def\rescale{0.2}

\begin{document}

\newcommand{\TNQF}{HTTN}

\title{Hamiltonian truncation tensor networks for quantum field theories}

\author{Philipp Schmoll}
\affiliation{\FUB}

\author{Jan Naumann}
\affiliation{\FUB}

\author{Alexander Nietner}
\affiliation{\FUB}

\author{Jens Eisert}
\affiliation{\FUB}
\affiliation{\HZB}

\author{Spyros Sotiriadis}
\affiliation{\UOC}
\affiliation{\FUB}

\date{\today}

\begin{abstract}

Understanding the equilibrium properties and out of equilibrium dynamics of quantum field theories are key aspects of fundamental problems in theoretical particle physics and cosmology. 
However, their classical simulation is highly challenging. 
In this work, we introduce a tensor network method for the classical simulation of continuous quantum field theories that is suitable for the study of low-energy eigenstates and out-of-equilibrium time evolution. 
The method is built on Hamiltonian truncation and tensor network techniques, bridging the gap between two successful approaches. 
One of the key developments is the exact construction of matrix product state representations of global projectors, crucial for the implementation of interacting theories. 
Despite featuring a relatively high computational effort, our method dramatically improves predictive precision compared to exact diagonalisation-based Hamiltonian truncation, allowing the study of so far unexplored parameter regimes and dynamical effects. 
We corroborate trust in the accuracy of the method by comparing it with exact theoretical results for ground state properties of the sine-Gordon model. 
We then proceed with discussing $(1+1)$-dimensional quantum electrodynamics, the massive Schwinger model, for which we accurately locate its critical point and study the growth and saturation of momentum-space entanglement in sudden quenches. 

\end{abstract}

\maketitle  

The discovery of \emph{tensor network} (TN) methods for the efficient classical simulation of quantum many-body models and their dynamics has revolutionised the study of condensed matter physics. 
Quantum many-body systems suffer from the infamous `curse of dimensionality' problem: their Hilbert space dimension increases exponentially fast with the number of particles, which makes the computational problem of finding their ground state or simulating their dynamics quickly intractable by exact methods. 
However, due to the short ranged nature of the interactions and low entanglement in typical condensed matter models, representing physically relevant eigenstates of the Hamiltonian of a model as \emph{matrix product states} (MPS)~\cite{Orus2019,RevModPhys.93.045003} (also called \emph{tensor trains}), 
i.e., chains of comparably small sized tensors to be contracted with each other, allows a dramatic reduction of the computational cost of their simulation.

A class of physically important models that have benefited to a lesser extent from the TN revolution yet are \emph{quantum field theory} (QFT) models. 
QFT is the foundation of particle physics and the basis upon which theories of the early universe dynamics or quantum information in black holes are built. 
QFTs are inherently many-body models, therefore suffering from the same `curse of dimensionality' problem as condensed matter models. 
But beyond that, their simulation poses an additional challenge: 
being continuous, they inevitably live in an infinite dimensional Hilbert space, meaning that numerical simulations would resort to some kind of truncation. This is even true for methods such as \emph{continuous matrix product states}~\cite{PhysRevLett.104.190405,PhysRevLett.105.260401,PhysRevX.9.021040} that readily incorporate the continuum in space. 
The validity and efficiency of a simulation technique relies strongly on the choice of truncation criterion, which must guarantee convergence to the correct target value and should ideally be generic. 

Previous attempts to classically simulate QFTs include instances of lattice discretisation~\cite{lattice_1b,PhysRevD.88.085030}, \emph{Hamiltonian truncation} 
(HT) methods~\cite{YUROV1990,Hogervorst2014}, and continuous tensor network methods~\cite{PhysRevLett.104.190405,PhysRevLett.105.260401,PhysRevX.9.021040}. 
Lattice discretisation has allowed the application of TN-based techniques available for lattice models (DMRG, TEBD, TDVP)~\cite{lattice_2, Banuls2019,PhysRevD.88.085030}. 
However, certain aspects of QFTs are fundamentally linked to their continuous nature making it harder, if not impossible, to simulate by means of lattice models. 
A notable artefact of lattice discretisation is the infamous fermion doubling problem~\cite{nielsen_no-go_1981}, due to which the discretised model has radically different symmetry properties compared to the original continuous QFT. 
These differences may be cured or neglected in the study of low energy physics (for example by explicitly breaking lattice symmetries or introducing non-local corrections), but may become significant in far from equilibrium problems. 
The introduction of continuous MPS demonstrated the capacity of TNs to accommodate continuous fields, however, its formulation makes it less straightforward to use for the simulation of QFT dynamics. 
On the other hand, HT methods operate directly on the continuum definition of QFT and are based on \emph{renormalisation group} (RG) theory to ensure truncation convergence~\cite{Hogervorst2014, james_non-perturbative_2018}. 
However, at the technical level they rely on exact diagonalisation or sparse matrix techniques and cannot address the crucial issue of exponential Hilbert space growth. 
As a result, their validity is limited to a small part of parameter space and their accuracy is relatively low. 
A naturally arising question is if the virtues of TNs can be combined with those of HT in a unified method and if so, how precisely this can be done. 

{\it Main aim of this work.}
To bridge the gap between two main approaches and to bring the best of both together, in this work, we introduce a TN method for the simulation of QFTs that we denote \TNQF. 
It is built upon ideas of HT but using \emph{matrix product operator} (MPO) representations of interactions in momentum space, thus avoiding the issues of lattice discretisation and at the same time reducing significantly the computational cost of simulation compared to exact diagonalisation. 

\begin{figure*}[ht]
    \begin{minipage}[b]{.35\textwidth}
        \begin{subfigure}{\textwidth}
            \subcaption{}
            \includegraphics[width = \linewidth]{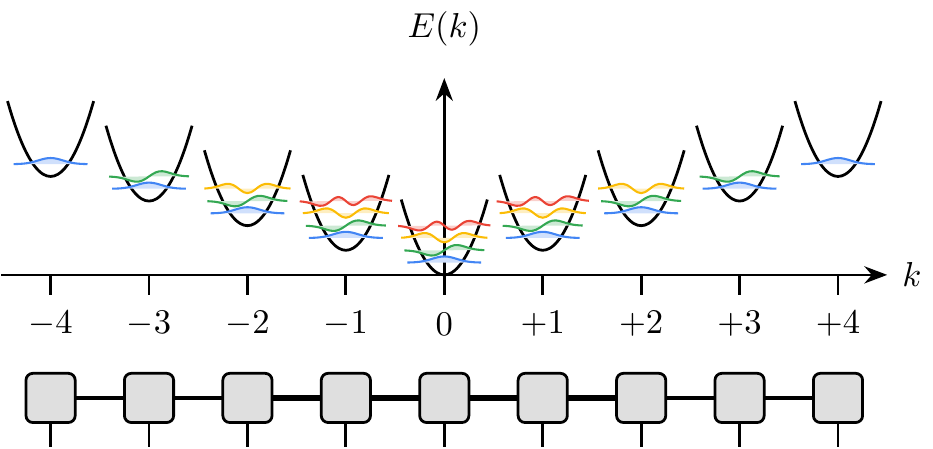}
            \label{fig:subfig11}
       \end{subfigure}
       \par
       \begin{subfigure}{\textwidth}
            \subcaption{}
            \includegraphics[width = \linewidth]{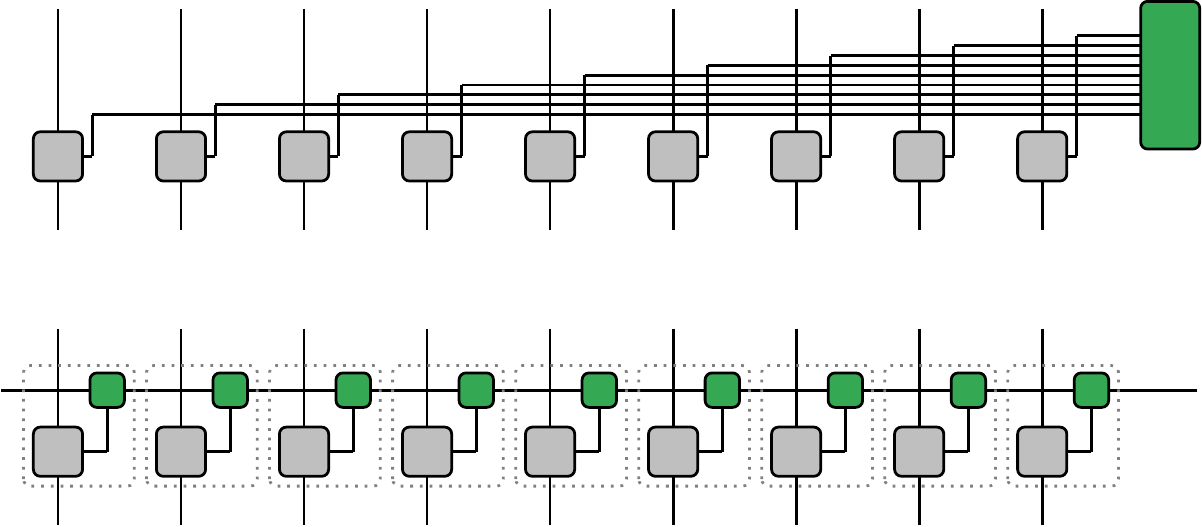}
            \label{fig:subfig12}
       \end{subfigure}
    \end{minipage}
    \qquad
    \begin{minipage}[t]{.55\textwidth}
        \begin{subfigure}{\textwidth}
            \subcaption{}
            \includegraphics[width = \linewidth]{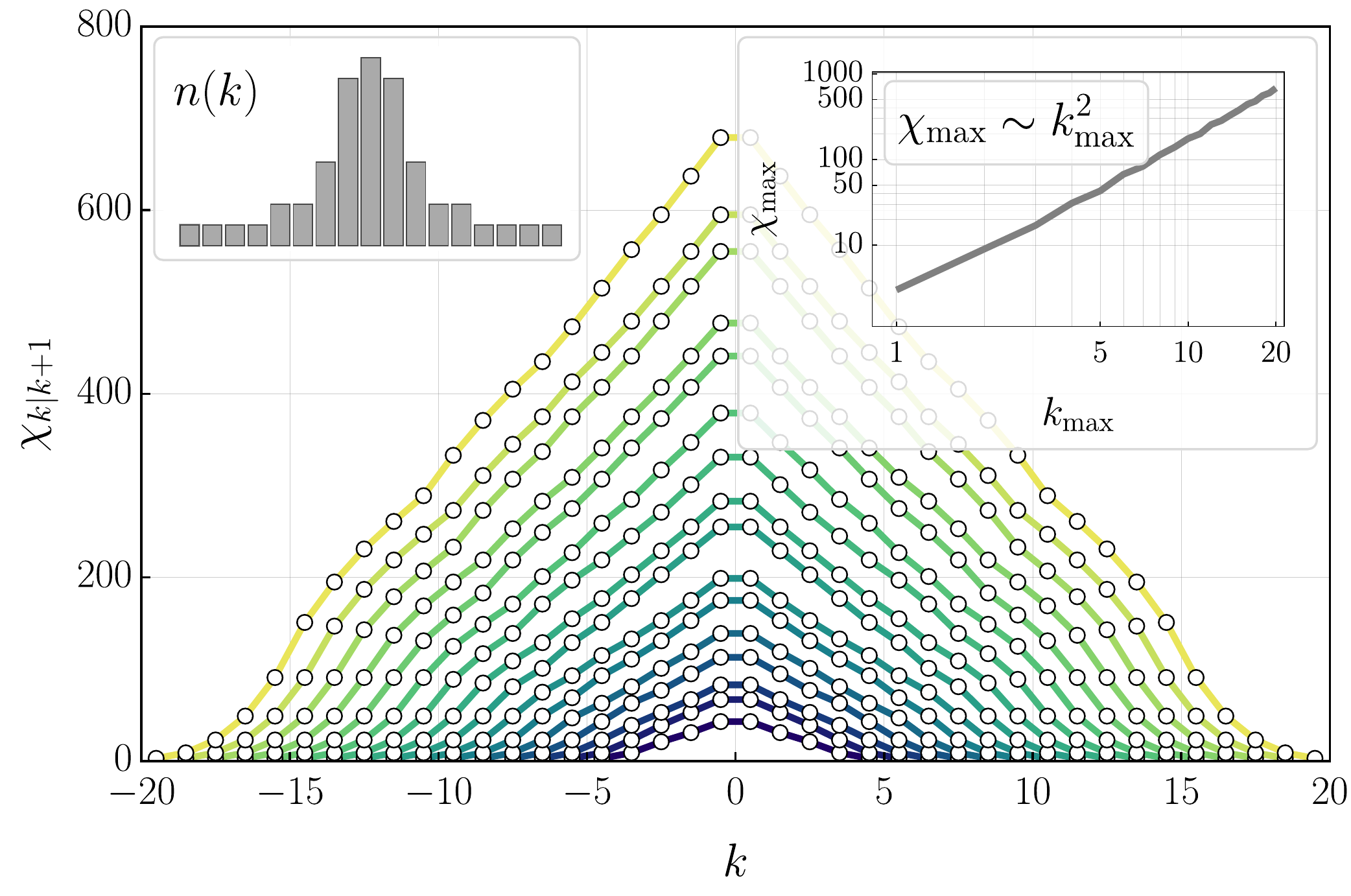}
            \label{fig:subfig13}
       \end{subfigure}
    \end{minipage}
    \caption{
    \textbf{Illustration of the tensor network.} 
    Following the ideas of HT, the target model's construction is built up on the basis of the free part of its Hamiltonian, therefore the Hilbert space is decomposed as the tensor product of those of the momentum modes, which correspond to independent harmonic oscillators (with the possible exception of the zero momentum mode). 
    (a) An MPS representation of quantum states is set up on this space, depicted as a chain of single mode tensors connected by horizontal links. 
    The non-contracted vertical tensor indices label the mode quantum numbers, which correspond to the Fock basis levels of the harmonic oscillators. 
    The contracted horizontal links (bonds) mediate quantum correlations between the modes and have variable dimension depending on the target state. 
    High-momentum modes are typically less occupied than low-momentum modes, prescribing a progressive truncation of the mode occupation levels. 
    (b) Graphical representation of the MPO representing the interacting part of the model's Hamiltonian. 
    The interaction potential is constructed out of products of imaginary exponential operators acting separately on each mode, represented by non-contracted single mode tensors (grey). 
    However, the conservation of total momentum imposes a global constraint on the mode quantum numbers that corresponds to contraction with a multi-index `Kronecker delta' tensor (green). 
    To obtain a regular MPO representation of the interaction operator, we recast this large tensor into an exact MPS form, which is then absorbed into the individual mode tensors. 
    (c) Scaling of the individual bond dimensions $\chi_{k|k+1}$ of the `Kronecker delta' MPS for various momentum cutoffs $k_\mathrm{max}$ and typical profiles of mode occupation $n(k)$ (left inset). 
    The maximal bond dimension $\chi_\mathrm{max}$ appears in the middle of the MPS and scales polynomially with $k_\mathrm{max}$ (right inset).}
    \label{fig:TN}
\end{figure*}

{\it Models considered as testing grounds.}
A testing ground for quantum many-body physics simulation methods is the class of $(1+1)$-dimensional models. 
This is related not only to the milder exponential growth of the Hilbert space size but also to the possibility of comparison with exact results provided by powerful analytical techniques available in certain $(1+1)d$ models, specifically, integrability~\cite{Zamo-Zamo_79} and bosonisation~\cite{Mandelstam,Coleman}. 
Two of the most interesting $(1+1)d$ QFT models are the \emph{sine-Gordon} (sG) and the \emph{massive Schwinger} (mS) model. 
The sG model is a bosonic QFT with a cosine self-interaction, defined by the Hamiltonian
\begin{align}
    H_\mathrm{sG} & = \int\left(\tfrac12 \Pi^{2} + \tfrac12 \left(\partial_{x}  \Phi\right)^{2} - \lambda \cos \beta \Phi \right) \, \mathrm{d} x
    \label{H_sG}
\end{align}
where $\Phi$ and $\Pi$ are canonically conjugate boson fields, and $\lambda$ and $\beta$ the interaction strength and frequency, respectively. 
The sG model is Bethe integrable~\cite{Zamo_sG_77,Zamo-Zamo_79} and reduces to the Dirac model of \emph{free fermions} (FF)~\cite{Mandelstam,Coleman} at $\beta = \sqrt{4\pi}$, which makes it ideal for testing the accuracy of simulation techniques. 
The mS model~\cite{schwinger_gauge_1962,coleman_more_1976} is the $(1+1)d$ version of \emph{quantum electrodynamics} (QED), perhaps the simplest model of a gauge theory, exhibiting confinement and string breaking. 
Originally defined in terms of a Dirac fermion field $\Psi$ of mass $m$ and electric charge $e$ interacting with the electromagnetic field $A_{\mu}$, it is equivalently expressed via bosonisation as a massive boson field with a cosine self-interaction~\cite{schwinger_gauge_1962,coleman_more_1976}
\begin{align}
  \begin{split}
    H_\mathrm{mS} & = \int\left(\tfrac12 \Pi^{2} + \tfrac12 \left(\partial_{x}  \Phi\right)^{2} + \tfrac12 M^2 \Phi^2 \right. \\ 
    & \qquad \qquad 
    \left. - \lambda \cos \left(\sqrt{4\pi} \Phi - \theta \right) \right) \, \mathrm{d} x,
    \label{H_mS}
  \end{split}
\end{align}
where the effective boson mass $M$ is proportional to the electric charge $e$ and the interaction strength $\lambda$ to the fermion mass $m$. 
The angular variable $\theta$ represents the background electric field. 

In HT, the free part of the model's Hamiltonians is used to define the Hilbert space basis and construct the interacting part. 
For the mS model the free part is a massive boson model, which in momentum space reduces to an infinite set of independent harmonic oscillator modes. 
For the sG model the free part is a massless boson model, which can be treated similarly, except for the zero momentum mode~\cite{SM}. 

{\it Setting up the method.}
Assuming periodic boundary conditions in a system of finite size $L$, the allowed momenta are all integer multiples of $2\pi/L$. 
Each momentum mode can be occupied by any number of bosons. 
The Hilbert space is therefore the tensor product of an infinite set of boson modes, labelled by the integer mode number $k$, each of which is spanned by states with any non-negative integer occupation number $n_k$. 
A natural truncation corresponds to restricting the maximum mode number $k_\mathrm{max}$, which plays the role of an \emph{ultra-violet} (UV) cutoff, and the maximum occupation numbers $n(k)$ for each of the modes. 
This space is suitable for setting up a TN method, where the truncated bosonic momentum modes play the role of TN `sites'. 
We model this system using the one-dimensional MPS, as shown in Fig.~\ref{fig:TN} (a). 
An MPS is a natural decomposition of the many-body state in this respect, where each harmonic oscillator with mode number $k$ and bosonic occupations $n_k = 0, 1, \dots, n(k)$ is represented by a single tensor. 
Entanglement in the system is mediated by virtual TN indices $\lbrace \alpha_k \rbrace$, which are the horizontal lines connecting adjacent sites in Fig.~\ref{fig:TN}(a). 
Concretely, the MPS setup in the truncated Hilbert space is then given by
\begin{align}
    \ket{\Psi_\mathrm{MPS}} = \sum_{\{n_k\}} \left( \sum_{\{\alpha_k\}}  \prod_{k} \, \lbrack \mathcal{M}^{(k)} \rbrack^{\alpha_{k+1}}_{\alpha_{k}, n_k} \right) \ket{\{n_k\}}
    \label{Psi_MPS}.
\end{align}
The site-dependent tensors $\mathcal M^{(k)}$ contain the variational parameters of the TN ansatz. 
Its accuracy is systematically controlled by the dimension of the virtual indices $\lbrace \alpha_k \rbrace$, i.e., by the so-called bond dimensions. 
Larger bond dimensions allow to capture more entanglement and enhance the MPS expressiveness, but also increase the computational cost and memory consumption of the representation.

The momentum modes are fully decoupled in the free model, however, the interaction couples them to each other in all possible ways. 
This is a crucial difference from typical TN applications in lattice models, where the interactions, being local, couple only neighbouring sites. 
Based on this observation we can already anticipate that the interaction-induced entanglement between the momentum modes may not be bounded in the same beneficial way as in local lattice models~\cite{Area_laws,CiracApproximability2008}. 
Given that entanglement is the limiting resource of TN methods, the number of maximal modes that can be simulated by our method is expected to be restricted much more than in lattice DMRG applications. 

Indeed, one step in the construction of the TN representation of the interacting Hamiltonian has a particularly high computational cost: 
The conservation of total momentum $P$, which enters as a condition imposed by the translation invariance symmetry of the interaction Hamiltonian. 
This condition translates into a combinatorial problem for the allowed interaction terms, which corresponds to a global symmetry constraint. 
In the framework of TNs, global symmetries can be implemented using \emph{symmetric tensor networks} (STN)~\cite{STN_0, STN}, a special type of TN that keeps track of the conserved charge associated with the symmetry at the level of each constituent tensor. 
Using STN we have analytically constructed exact TN representations of projectors onto global symmetry sectors, based on which the interacting Hamiltonian can be finally written in MPO form (see Fig.~\ref{fig:TN}(b)). 
As anticipated above, however, the MPO bond dimension increases quickly (polynomially) with the maximum mode number $k_\mathrm{max}$, as shown in Fig.~\ref{fig:TN}(c).\\

\begin{figure}
    \captionsetup[subfigure]{font=small, justification=justified, singlelinecheck = false, skip=0pt}
    \centering
    \subcaptionbox{}
    {
    \includegraphics[width = .95\columnwidth, valign = t]{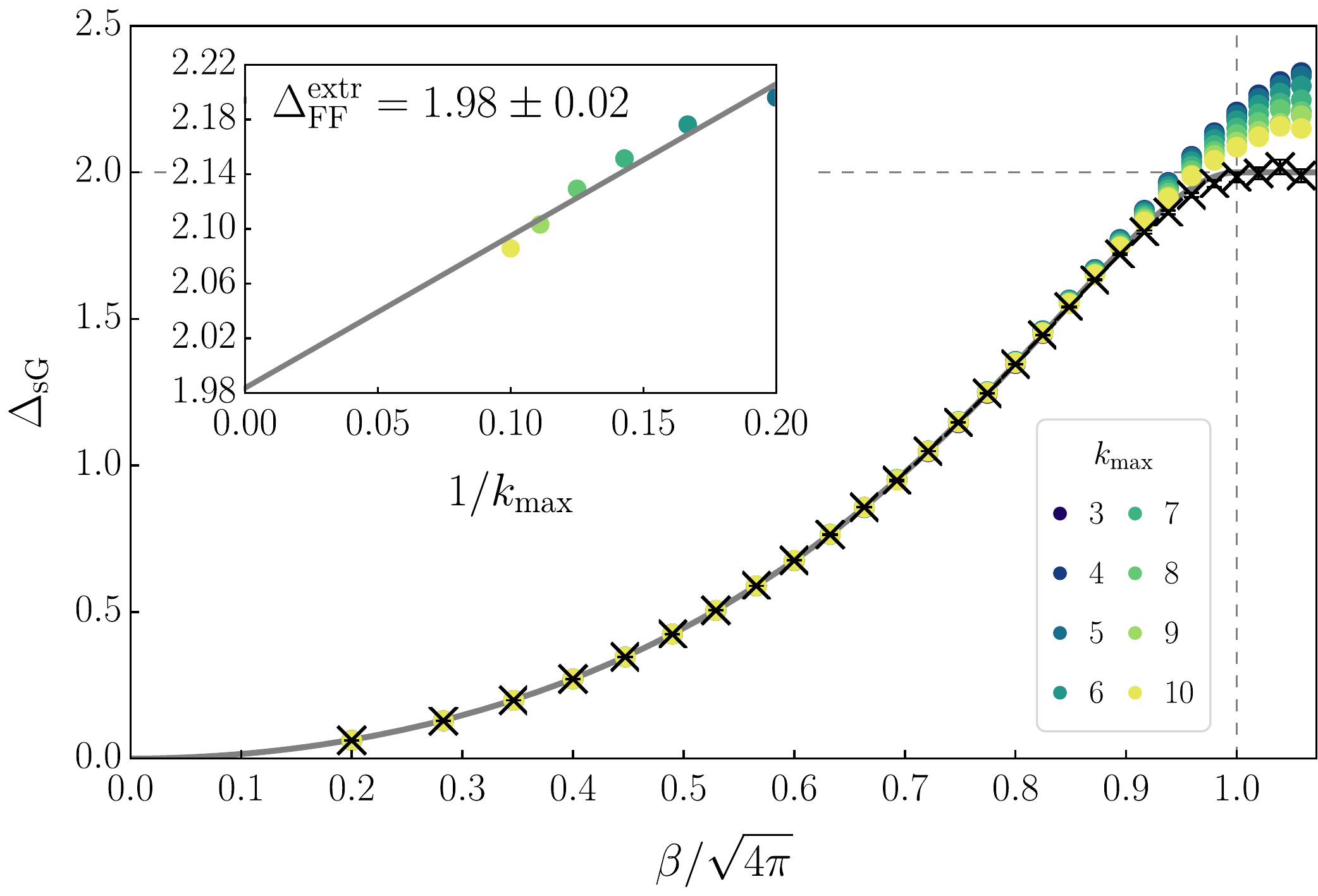} \label{fig:subfig21}
    }
    \subcaptionbox{}
    {
    \includegraphics[width = .95\columnwidth, valign = t]{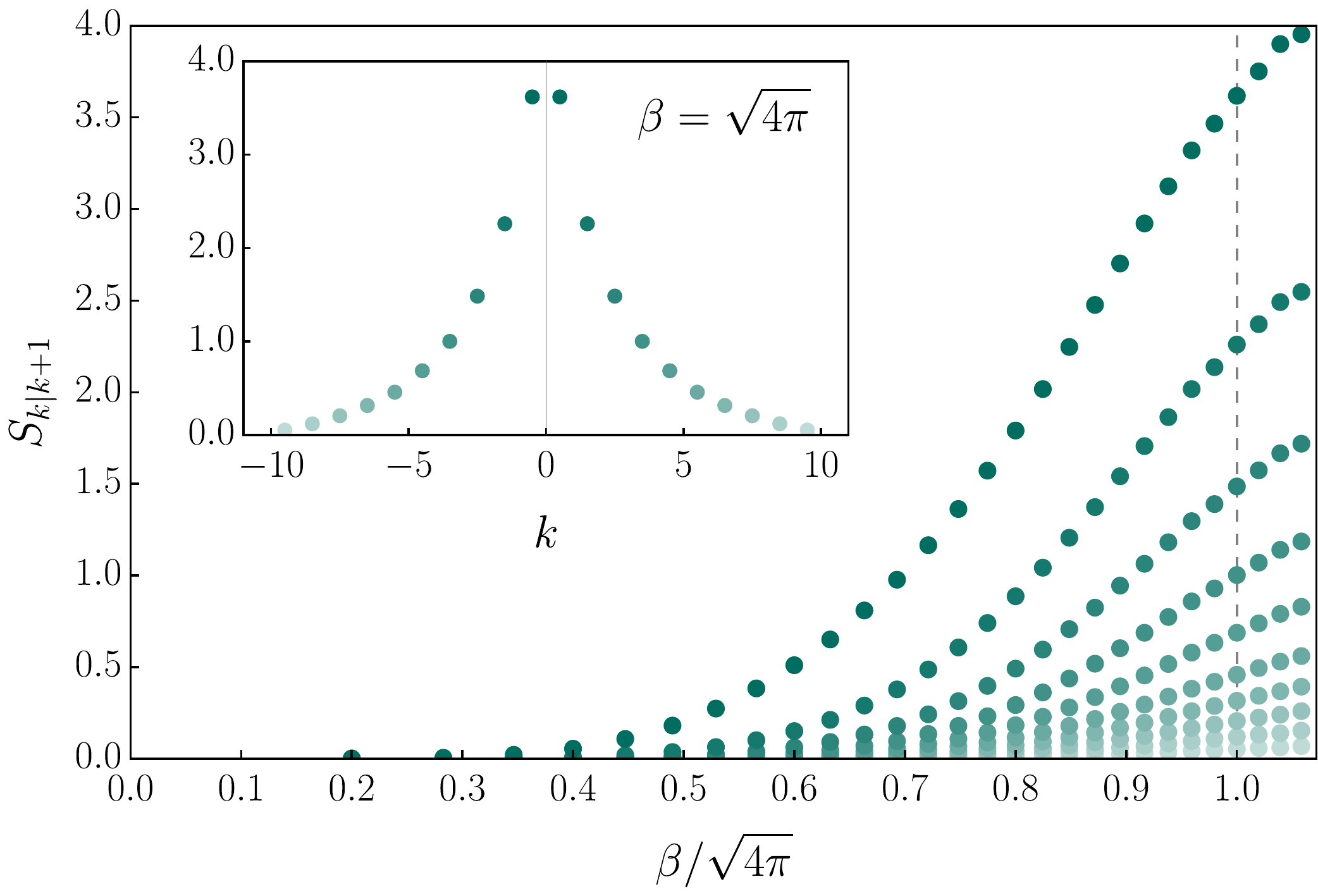}
        \label{fig:subfig22}
    }
    \caption{
    \textbf{Ground state simulations of the sine-Gordon model.} 
    (a) Energy gap as a function of the frequency $\beta$ of the cosine interaction. 
    We compute the energy gap as the energy difference between the ground and first excited state in the thermodynamic limit, at different cutoffs $k_\mathrm{max}$ for a sufficiently large finite volume ($L = 15$). 
    Finite size corrections decay exponentially with $L$ and are negligible at this size, whereas corrections due to the finite momentum cutoff are significant, especially at the FF point and beyond in the repulsive regime, $\beta \geq \sqrt{4\pi}$. 
    The theoretical prediction (grey solid line) is based on integrability~\cite{ZAMOLODCHIKOV1995}. 
    For $\beta \geq \sqrt{4\pi}$ the exact ground state energy diverges with the momentum cutoff but the energy gap is finite and exactly equal to $\Delta_\mathrm{sG} = 2$. 
    Due to this UV divergence a numerical estimation of the energy gap is only possible through extrapolation. 
    Our numerical estimates (`$\cross$' markers) are derived by extrapolation of the results to the limit $k_\mathrm{max}\to\infty$. 
    The numerical predictions so-derived agree very well with the theoretical ones, not only at the FF point but even beyond that. 
    The inset shows the extrapolation at the FF point based on a linear fit in $1/k_\mathrm{max}$~\cite{SM}. 
    (b) Entanglement entropy $S_{k|k+1}$ for successive bi-partitions in the MPS as a function of the frequency $\beta$ (at $k_\mathrm{max} = 10$). 
    The growth of $S_{k|k+1}$ with $\beta$ reflects the increasing complexity of the simulations. 
    The inset shows the entanglement entropy for all possible MPS cuts at the FF point.}
    \label{fig:sG}
\end{figure}

{\it Results for quantum field simulations.} 
Having constructed MPO representations of the Hamiltonians of the two {QFT} models enables the use of standard TN techniques for their simulation. 
We first study equilibrium properties. 
For the sG model the ground state energy density $\varepsilon_\mathrm{sG}$ and energy gap $\Delta_\mathrm{sG}$ are analytically known functions of $\beta$ in the thermodynamic limit~\cite{ZAMOLODCHIKOV1995}. 
Small $\beta$ or $L$ are easily accessible by standard HT but approaching the FF limit $\beta\to\sqrt{4\pi}$ in combination with $L\to\infty$ becomes increasingly difficult~\cite{TCSA-sG1,TCSA-sG2}. 
At the FF point and beyond, $\varepsilon_\mathrm{sG}$ diverges but energy level differences converge, so that an estimation of the energy gap is still possible. 
For the mS model analytical results are only available at the massless fermion point $m=0$, which is mapped to the free massive boson model. 
One of the most studied aspects of the model is its critical behaviour: varying $m$ along the line $\theta=\pi$, a second order phase transition occurs at the critical point $m=m_\mathrm{c}$ and for $m>m_\mathrm{c}$ chiral symmetry is spontaneously broken. 
A HT study of the mS model has been successfully carried out in Ref.~\cite{Kukuljan_MS}, however, convergence was achieved only for small values of $m$ well below the critical point. 
An accurate computation of the sG energy gap at the FF point $\Delta_\mathrm{FF}$ and of the mS critical point $m_\mathrm{c}$ are therefore two milestones for our method. 

As shown in Fig.~\ref{fig:sG} and Fig.~\ref{fig:mS}, these objectives can be easily achieved even at relatively low truncation cutoffs $k_\mathrm{max}$. 
Our best estimates are based on extrapolation of $k_\mathrm{max}\to\infty$. 
Beyond that we can successfully compute ground state expectation values of local (in coordinate space) observables, which provide valuable information about the two models, e.g., the mS phase diagram as shown in Fig.~\ref{fig:mS}. 
Naturally, as a byproduct of the simulations, our method provides also results for the \emph{momentum-space entanglement entropy}~\cite{momentum-space_EE_QFT}. 
Its dependence on the partitioning shows that it increases rapidly when the splitting position approaches the middle of the MPS chain, as shown in Fig.~\ref{fig:sG} for the sG model. 
\begin{figure*}[p]
    \captionsetup[subfigure]{font=small, justification=justified, singlelinecheck = false, skip=0pt}
    \centering
    \begin{tabular}{@{}c@{}c@{}c@{}}
    \subcaptionbox*{(a)}
    {
        \includegraphics[height = \rescale\paperheight, valign=t]{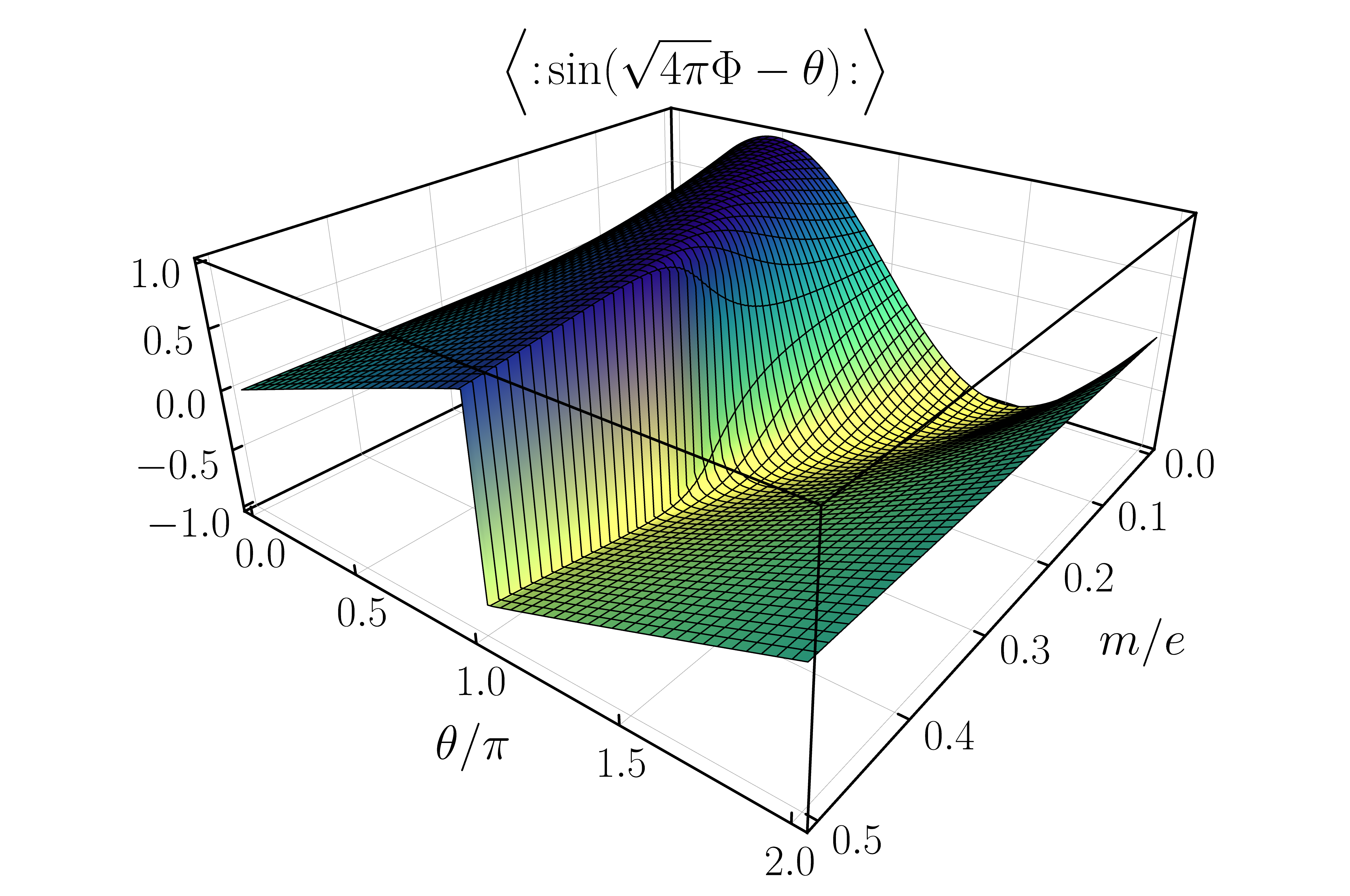}
    }
    & \hspace{2cm} &
    \subcaptionbox*{(d)}
    {
        \includegraphics[height = \rescale\paperheight]{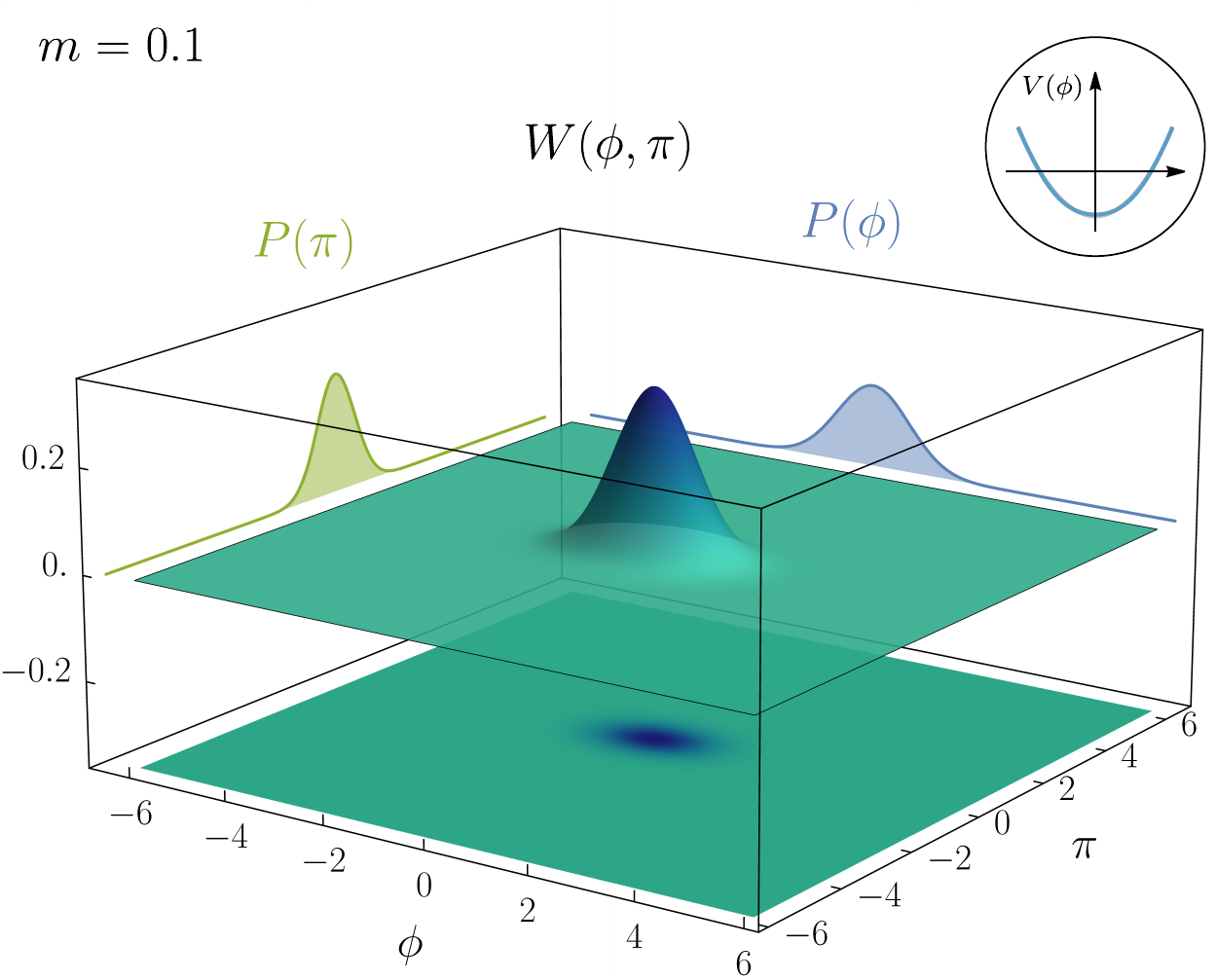}
    } \\
    & \\[-1.88ex]
    \subcaptionbox*{(b)}
    {
        \includegraphics[height = \rescale\paperheight]{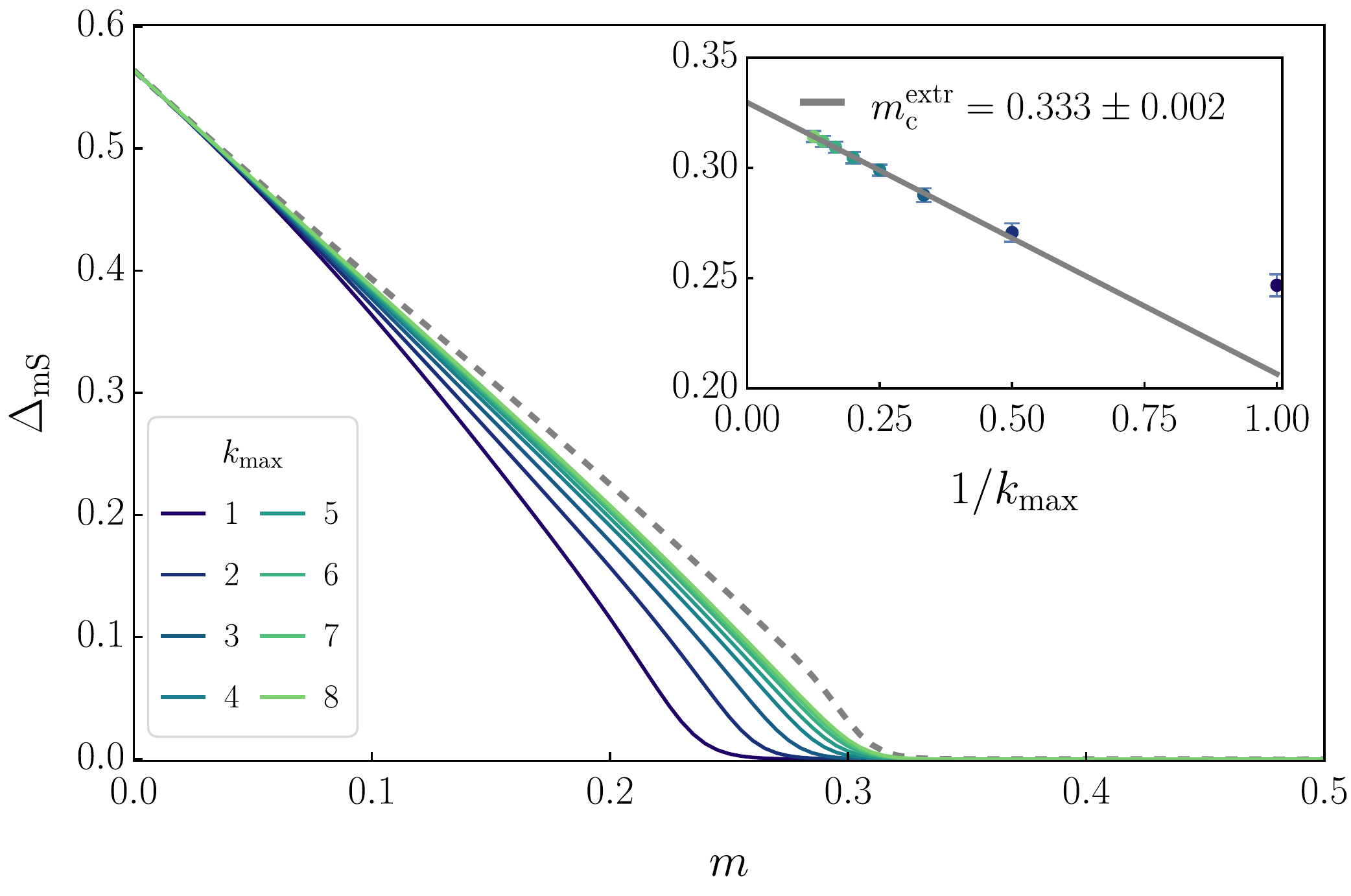}
    } 
    &  &
    \subcaptionbox*{(e)}
    {
        \includegraphics[height = \rescale\paperheight]{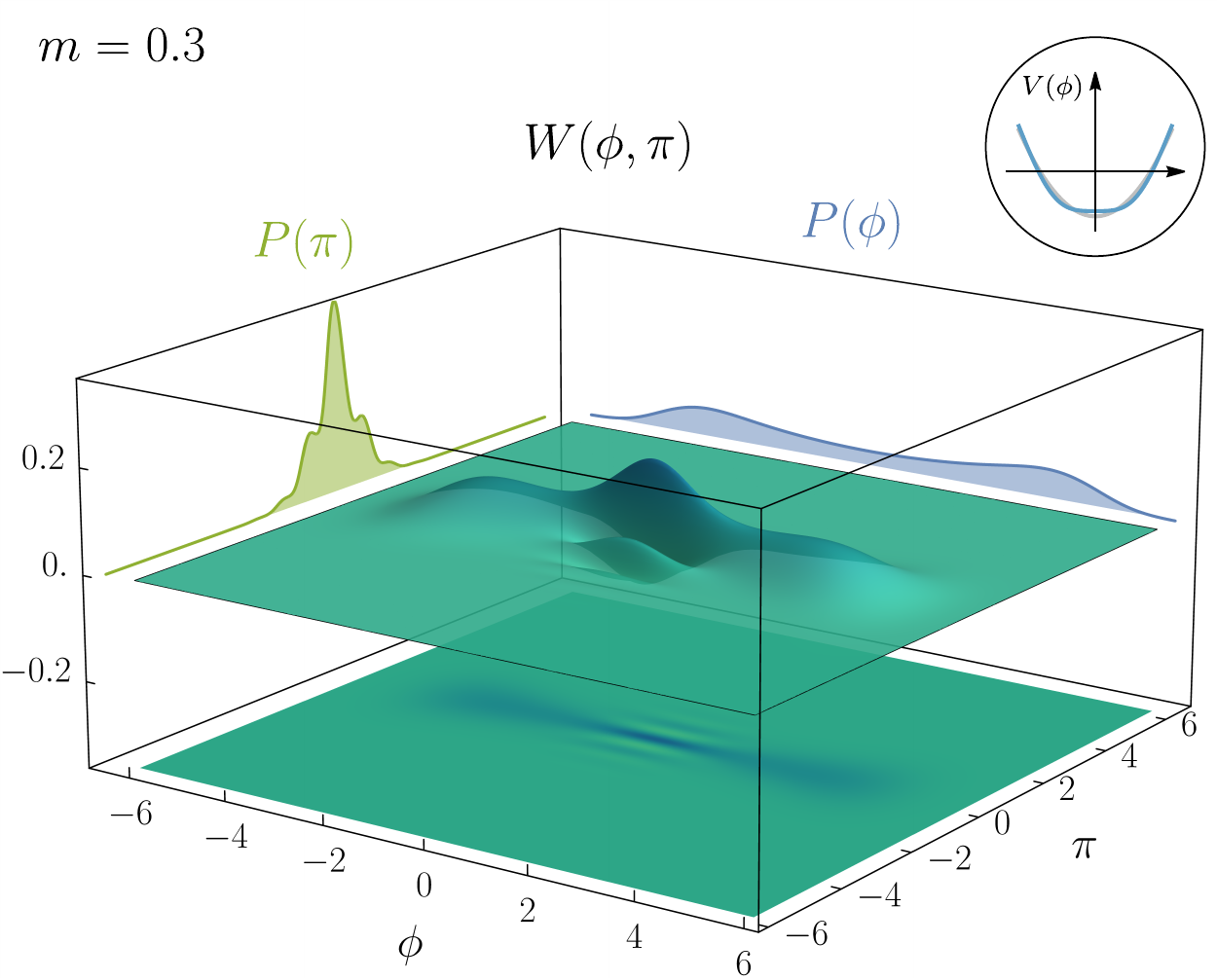}
    } \\
    & \\[-1.88ex]
    \subcaptionbox*{(c)}
    {
        \includegraphics[height = \rescale\paperheight]{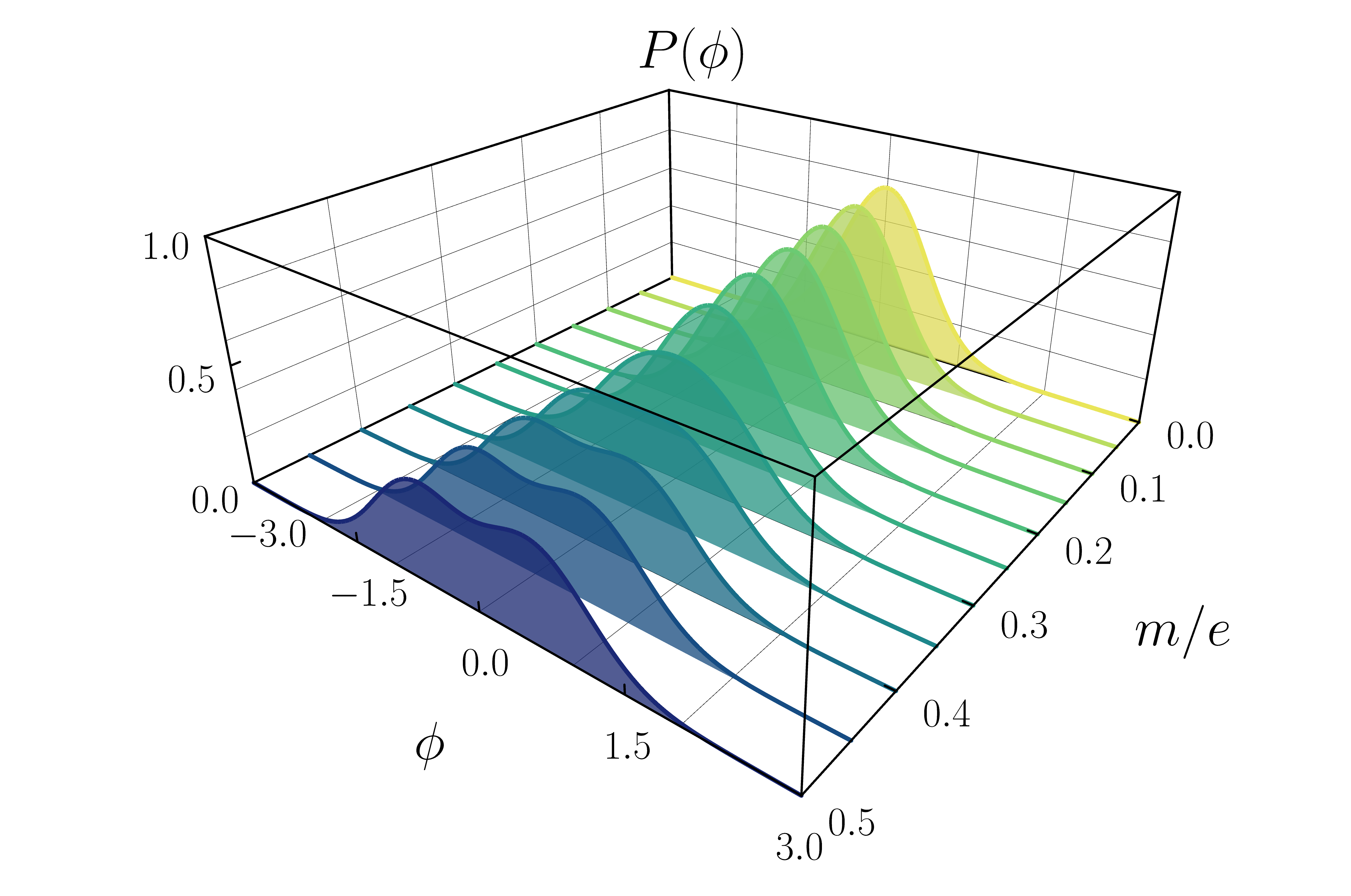}
    } 
    &  &
    \subcaptionbox*{(f)}
    {
    \includegraphics[height = \rescale\paperheight]{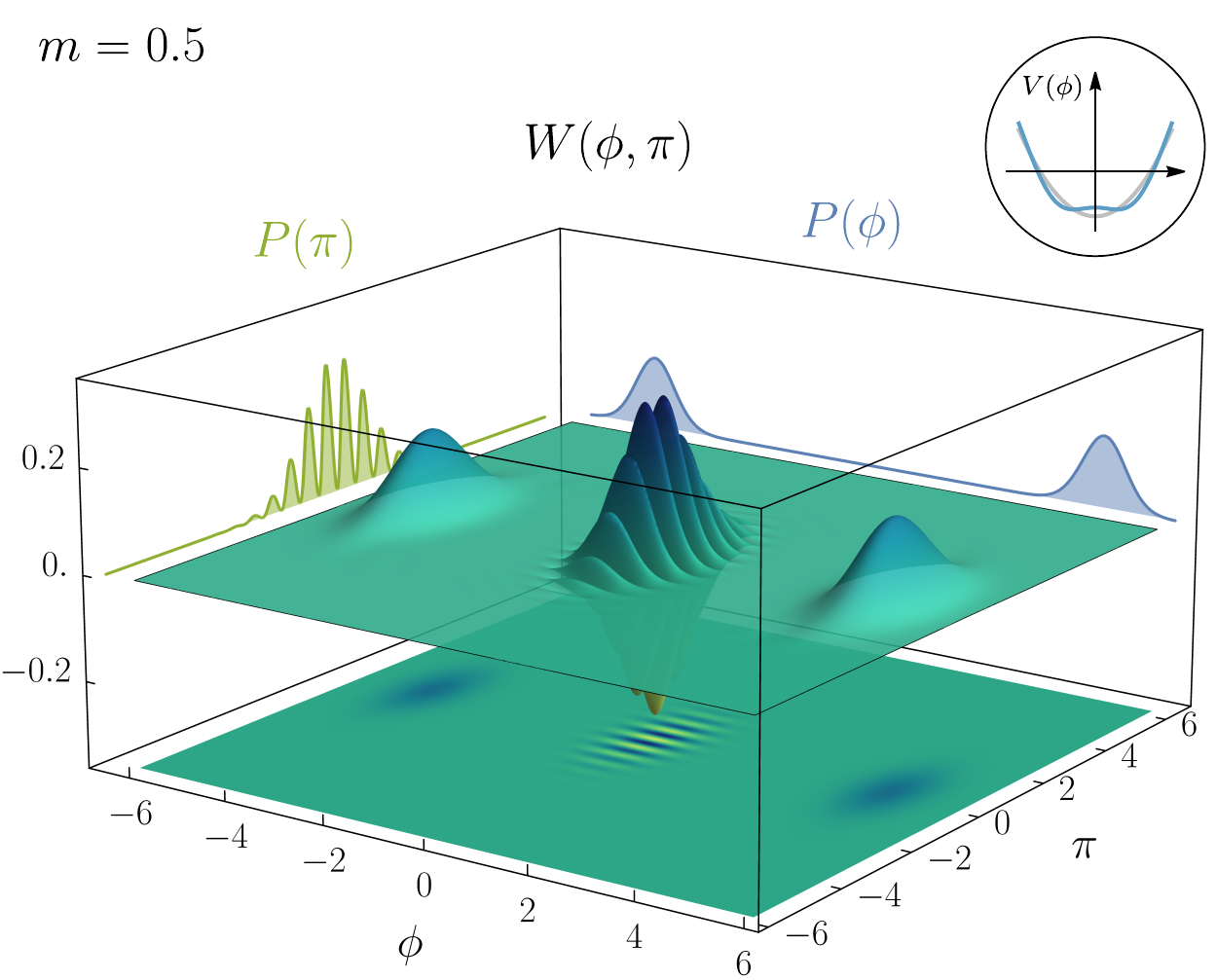}
    } 
    \end{tabular}
    \caption{\textbf{Ground state simulations of the massive Schwinger model.} 
    (a) Order parameter (ground state expectation value of $\twoOrd{\sin({\sqrt{4\pi}\Phi - \theta})}$, that is proportional to $\braket{\bar\Psi \gamma_5 \Psi}$~\cite{iso_hamiltonian_1990}) as a function of the fermion mass $m$ and the background electric field parameter $\theta$ at $k_\text{max} = 6$. 
    The phase transition across the line $\theta=\pi, m>m_c$ is indicated by the jump in the order parameter. 
    (b) Energy gap as a function of $m$ for different cutoffs $k_\mathrm{max}$ and estimation of the critical mass by extrapolation. 
    Our numerical result $m_\text{c}=0.333(2)$ agrees well with the earlier numerical estimate $m_\text{c} = 0.3335(2)$~\cite{mS_phase_transition}. 
    (c) Probability distribution (full counting statistics) of $\Phi$ in the ground state for various values of $m$ and (d) - (f) Wigner quasi-probability distribution $(\phi,\pi)\mapsto W(\phi,\pi)$ and corresponding marginal distributions $\phi\mapsto P(\phi)$ and $\pi\mapsto P(\pi)$ for the reduced density matrix of the zero momentum mode at $m/e=0.1$ (d), $0.3$ (e) and $0.5$ (f). 
    The insets show the semiclassical effective potential. 
    Below the critical point the Wigner function is close to a squeezed Gaussian distribution, whereas above it the presence of two peaks and the characteristic alternating sign features in-between indicate a quantum superposition of two degenerate states centred at the minima of the potential.}
    \label{fig:mS}
\end{figure*}

{\it Venturing into non-equilibrium problems.} 
More importantly, the \TNQF\ method can be applied in a straightforward way to the simulation of non-equilibrium dynamics. 
As an example, let us consider a quench of the fermion mass in the mS model starting from $m_0=0$, as shown in Fig.~\ref{fig:mS_quench_dynamics}. 
In the bosonised representation, this quench corresponds to starting from the free boson ground state and switching on the interaction. 
In this case momentum-space entanglement is zero in the initial state and subsequently grows with time. 
Quenching to $m<m_\text{c}$ results in oscillations typical of a mass quench~\cite{Calabrese_2007}, whereas for $m>m_\text{c}$, which corresponds to quenching across the critical point, a fast relaxation appears to emerge. 
This is visible in the dynamics of both expectation values of local observables and momentum-space entanglement entropies. 
The observation of relaxation in the simulation is nontrivial and relies crucially on the capability of the method to reach high truncated space dimensions. 

\begin{figure*}
    \centering
    \begin{subfigure}{0.45\textwidth}
        \subcaption{}
        \includegraphics[width = \linewidth]{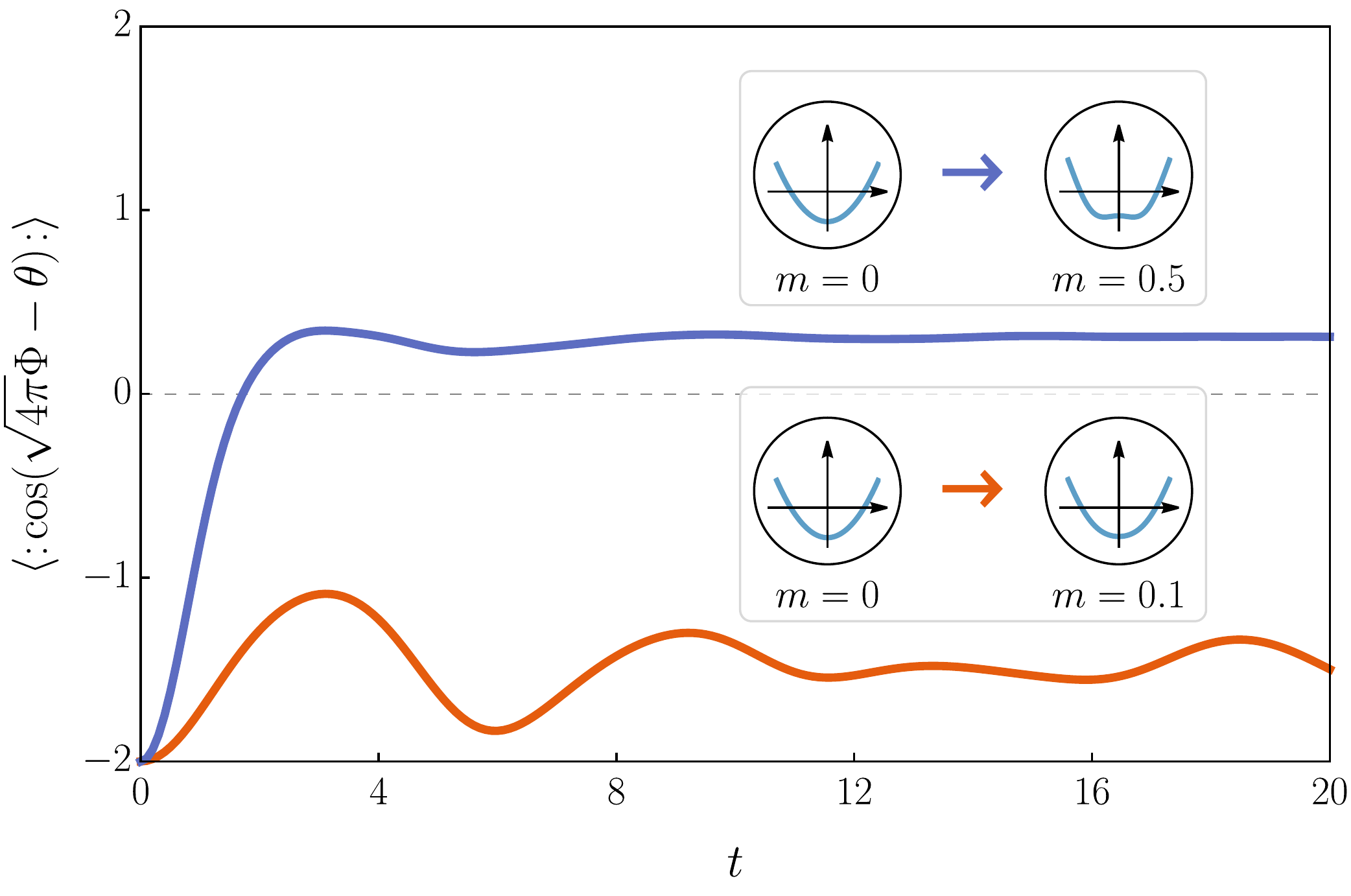}
        \label{fig:subfig41}
    \end{subfigure}
    \qquad
    \begin{subfigure}{0.45\textwidth}
        \subcaption{}
        \includegraphics[width = \linewidth]{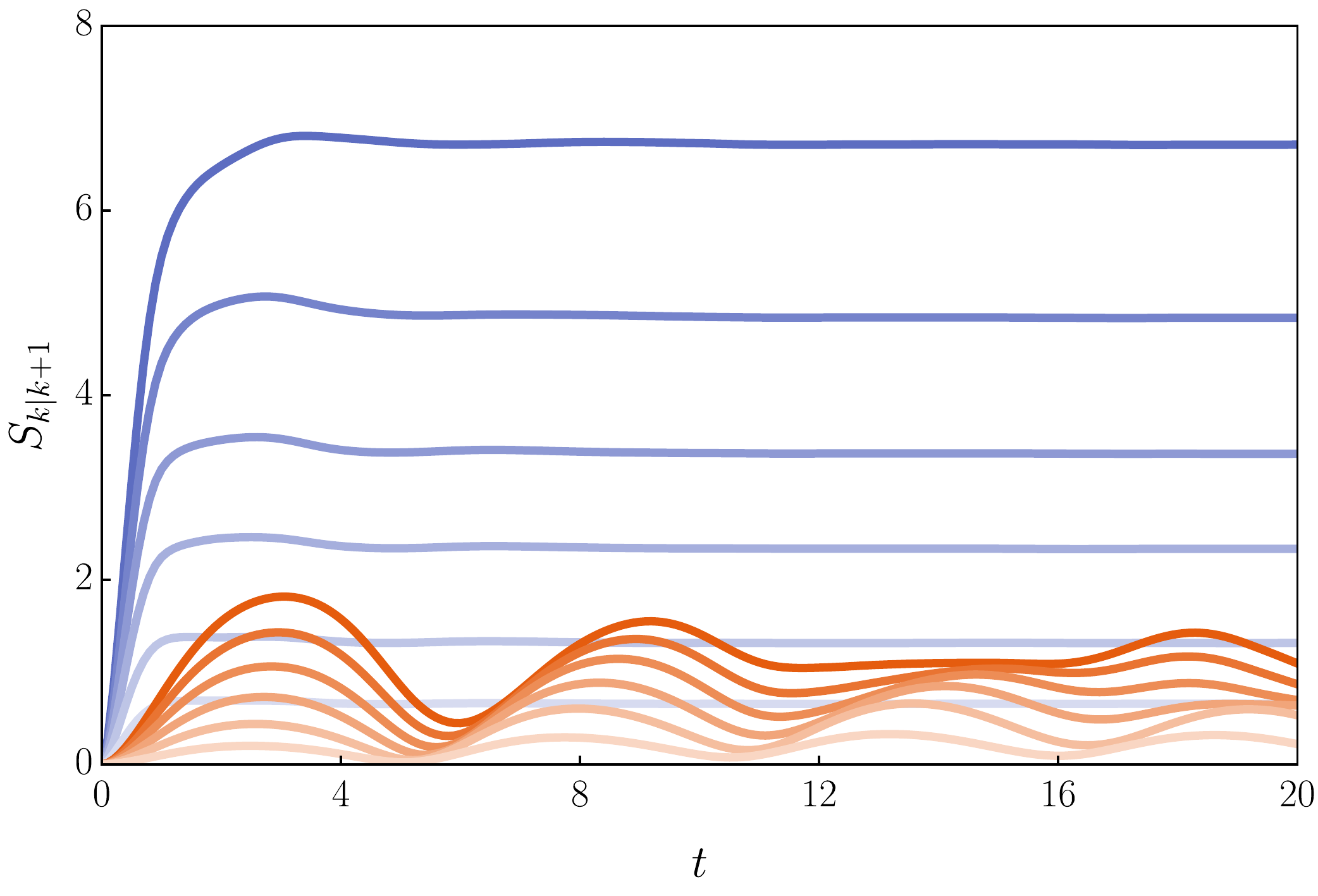}
        \label{fig:subfig42}
    \end{subfigure}
    \begin{subfigure}{0.95\textwidth}
        \subcaption{}
        \includegraphics[width = \linewidth]{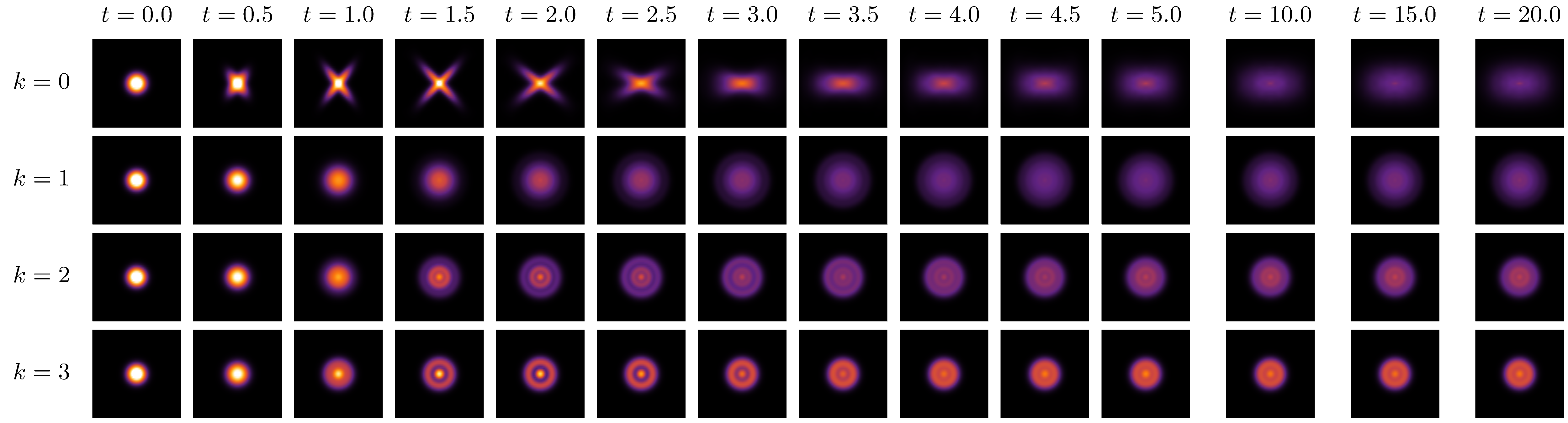}
        \label{fig:subfig43}
    \end{subfigure}
    \caption{\textbf{Simulations of out-of-equilibrium dynamics in the massive Schwinger model.} 
    Quench of the fermion mass from \mbox{$m_0 = 0$} to $m = 0.1e < m_c$ (case A) or $m = 0.5e > m_c$ (case B) at fixed $\theta=\pi$ and $k_\text{max} = 6$. 
    In the semiclassical approximation, the post-quench potential in case A is a weak perturbation of the pre-quench one and has only one minimum, whereas in case B it has two minima located at opposite values. 
    The initial state is a Gaussian centred at the middle. 
    (a) The subsequent dynamics of $\langle\twoOrd{\cos(\sqrt{4\pi}\Phi - \theta)}\rangle$ is markedly different in the two cases: instead of the oscillatory behaviour of case A, characteristic of a typical mass quench, in case B we observe a rapid approach to a limiting value, indicating relaxation. 
    (b) Entanglement entropy $S_{k|k+1}$ for successive bi-partitions in the MPS computed for $k = -6, -5 \hdots, -1$ in the two types of quench, distinguished by light to dark shades of orange (A) and blue (B). 
    The behaviour is similar to the cosine expectation value and in case B the entropy shows almost full saturation. 
    (c) Snapshots of the Wigner functions of the lowest modes for quench B at various different times reveal the complex underlying dynamics and the onset of relaxation.}
    \label{fig:mS_quench_dynamics}
\end{figure*}

{\it Full counting statistics.}
In addition to computing observables, our method can provide a detailed characterisation of the computed quantum states. 
The probability distribution of observables or what is called \emph{full counting statistics} {(FCS)} encodes valuable information about their fluctuations in a given state more transparently than expectation values and higher moments. 
Moreover, those quantities are directly accessible in and comparable to experimental measurements~\cite{FCS_exp_1,FCS_exp_2}. 
The Wigner quasi-probability function of each mode provides a phase-space representation of the state of a system, illustrating its quantum nature and contrasting it with semiclassical approximations. 
Knowledge of the initial state's Wigner functions can be used as input in simulations of the corresponding classical dynamics via the truncated Wigner approximation~\cite{truncated-Wigner_Polkovnikov}, allowing a precise comparison of classical and quantum dynamics in a given model. 
Fig.~\ref{fig:mS}(c) shows how increasing the fermion mass $m$ beyond $m_\mathrm{c}$ results in a bimodal shape for the FCS of the local field $\Phi
$. 
Studying the Wigner functions reveals that this is mainly due to the zero momentum mode adopting a highly non-Gaussian distribution. 

{\it Comparison with standard truncation methods.} 
The above results show that \TNQF\ is a significant improvement over 
the traditional exact diagonalisation-based HT method. 
This is especially clear in the mS model, not only from the accuracy of ground state results already at relatively low cutoffs and the extension of the accessible parameter space beyond earlier limitations, but also from the achievement of observation of the steady state after a quench. 
These improvements are an indication that the large number of higher occupied basis states that are included in the mode-truncated variational space of \TNQF\ contain relevant information that is inaccessible by the energy-cutoff based truncation of standard HT. 
On the other hand, exact diagonalisation limitations do not apply to variants of HT based on sparse matrix techniques~\cite{Rychkov2014} or exploiting the block structure of the Hamiltonian matrices~\cite{Horvath2022_CFTCSA}, and it would be interesting to compare our method with those in terms of accuracy and performance. 

One advantage of the \TNQF\ method is that the system size $L$ appears merely as a parameter in the MPO representation of the Hamiltonians and increasing it translates at worst to a relatively mild increase in the bond dimension both for ground state and dynamics simulations. 
For this reason, studying the thermodynamic limit does not pose a significant challenge for our method. 
On the contrary, increasing the cutoff $k_\mathrm{max}$ is far more challenging due to the rapidly growing bond dimension of the Hamiltonian MPO. 
Even at the relatively low cutoffs used in our simulations the maximum MPS bond dimension reaches values of the order of a thousand, much higher than in simulations of local lattice models but still moderate compared to TN applications in quantum chemistry. 
Nevertheless, increasing $k_\mathrm{max}$ beyond current levels is essential for certain physical applications (e.g., problems that involve spatial inhomogeneity), therefore future improvements of the method should focus on this bottleneck. 

{\it Potential for using mode transformations.}
Luckily, there is a large number of options for handling this problem. 
One of the advantages of the momentum-space formulation is the possibility of a low-cost optimisation of the computational basis by means of suitable bosonic \emph{mode transformations}~\cite{basis-optimisation_nonlocal_2,PhysRevB.104.075137}, which are known to improve tensor network methods for particle Hamiltonians. 
Assuming translation invariance, a Gaussian approximation of a quantum state can only have correlations between opposite momentum modes. 
Therefore, a significant reduction in computational effort can be achieved by reordering the modes into a sequence of opposite momentum pairs and using two-mode unitary transformations to switch to the basis where the target state's occupations are minimal. 
The number of optimisation parameters to be determined scales only linearly with the number of modes, hence the optimisation search is tractable. 
In view of the above, we anticipate that momentum-space entanglement is only moderate whenever a Gaussian approximation is satisfactory. 
In the context of out of equilibrium dynamics, in particular, this condition guarantees that momentum-space entanglement remains low at all times, enabling indefinitely long simulation times. 
Conversely, momentum-space entanglement is high, possibly unbounded, in the case of strongly correlated states that do not admit a Gaussian approximation. 
In our mS simulations this interpretation is strongly suggested by the relatively high entanglement entropy at large times after quenching across the critical point. 

{\it Perspectives for MPO compression.}
Apart from mode reordering and mode transformation optimisation, alternative options for improving \TNQF's efficiency include the exploitation of additional symmetries ($Z_2$ field parity symmetry), use of different types of TNs (tree or other types of graphs)~\cite{tree-TN,TNbook} and numerical compression of the Hamiltonian MPO by 
means of singular value decomposition~\cite{MPOcompression1,MPOcompression2}. 
These are already developed TN techniques and can be easily incorporated here, which is an additional advantage of \TNQF. 
 
{\it Comparison with lattice DMRG.}
Comparing \TNQF\ with lattice-based DMRG simulations in terms of accuracy of ground state simulations, the outcome varies significantly depending on the model representation (bosonic or fermionic) and parameter values (cf.~Ref.~\cite{Roy2021} for sG and Refs.~\cite{mS_phase_transition,mS_mps_Banuls_2013,mS_mps_Orus_2017} for mS). 
Moreover, given the possibility of controlled extrapolations to the limit $k_\text{max} \rightarrow \infty$, our method seems to be of similar accuracy as relativistic continuous MPS without a UV cutoff~\cite{Tilloy2021}. 
However, the major potential of \TNQF\ compared to lattice DMRG simulations lies in the simulation of out of equilibrium dynamics. 
While lattice simulations suffer from unbounded entanglement growth~\cite{buyens_MS_real-time_2017}, which limits the maximal accessible time, this limitation seems to be absent in our approach, at least for certain types of quenches (cf. Fig.~\ref{fig:mS_quench_dynamics}). 

{\it Comparison with conventional Hamiltonian
truncation methods.} 
The TN implementation of HT offers a new perspective on the problem of classical simulation of QFTs. 
The method dramatically expands the accessible size of the truncated Hilbert space far beyond the reach of exact diagonalisation, leading to a significant improvement in accuracy compared to standard HT and opening up the possibility to study the emergence of equilibrium in QFTs. 
At the conceptual level, the coalition of HT and TN methods can provide a quantum information theoretic view on the traditional RG approach to QFT. 
The decoupling of low- and high-wavelength degrees of freedom in non-critical QFT models can be given a precise meaning in terms of momentum-space entanglement. 

The representation of Hamiltonians in terms of TNs has the additional substantial advantage that the MPO formalism can be readily used to simulate QFTs at thermal equilibrium. 
In combination with the observation of equilibration in simulations, this would enable direct tests of ergodicity in QFTs~\cite{QFT_ergodicity_1, QFT_ergodicity_2} through the comparison of steady states expectation values with those at thermal equilibrium. 
Unsurprisingly, the computational cost of thermal state simulations is even higher than that of the pure states algorithms presented above, so we leave this question for future investigation.

{\it Perspectives for higher dimensional quantum field theories.} 
Even though we have focused here on two paradigmatic $(1+1)d$ models, the method can be extended to any other model for which HT is applicable, provided that the corresponding unperturbed model is originally defined in a tensor-product space. 
This class of models includes the $\phi^4$ model~\cite{Rychkov2014,Bajnok_16}, $(1+1)d$~QCD~\cite{HT_QCD}, and potentially also higher dimensional models~\cite{Hogervorst2014} using \emph{projected entangled pair states} (PEPS), a generalisation of MPS to higher dimensions~\cite{PEPS_2}. 
From a technical perspective, one of the main achievements in the development of the method is the exact construction of MPS representations of global projectors. 
This extends the application of TNs to problems defined in the solution space of a global constraint. 

{\it Summary and outlook.}
We have introduced a new Hamiltonian truncation tensor network method for continuous quantum field theories, which combines two standard methods from different worlds to capture quantum field theories classically both in and out of equilibrium. 
This method is a significant improvement over the state of the art for this purpose. 
We hope that the present work will contribute to rendering tensor methods more common tools for studying interacting quantum field theories and their dynamics.

\section*{Methods}
\label{sec:methods}

The \TNQF\ method presented in this work is based on classical simulation of quantum many-body systems using \emph{tensor networks} (TN). 
TNs offer efficient representations of the full many-body wave function based on a network of interconnected local tensors, whose contraction recovers the probability amplitudes. 
In one spatial dimension the most prominent family of TNs is known as the \emph{matrix product state} (MPS), upon which the \TNQF\ method is built. 
This enables the immediate application of established and standard numerical TN methods, such as the \emph{density matrix renormalization group} (DMRG)~\cite{White1992} for the simulation of ground and low-energy excited states or the \emph{time dependent variational principle} (TDVP)~\cite{Haegeman2011} for the simulation of dynamics, with little modifications. 
Before we present the \TNQF\ method in detail, we introduce the basic ideas of HT. 

\subsection*{Hamiltonian truncation}

HT is based on the following idea: 
A general interacting QFT model can be considered as a perturbation of a suitably chosen exactly solvable, e.g., free, model. 
Using the latter to construct the Hilbert space basis and truncating it based on some criterion, e.g., a maximum energy cutoff, one can obtain matrix approximations of the target model's Hamiltonian by projecting the perturbation operator onto the truncated finite subspace. 
Specifically, if $H$ is the target model's Hamiltonian we split it as a sum of the free part $H_0$ and the interaction $V$ with a coefficient or `coupling constant' $\lambda$, 
\begin{align}
    H = H_0 + \lambda V.
    \label{HT}
\end{align}
Let us consider a truncation of the Hilbert space spanned by all eigenstates of $H_0$ that have eigenvalue $E\leq E_\text{cut}$ for some cutoff value $E_\text{cut}$. 
We construct a sequence of matrix approximations of $H$ in the truncated subspaces that correspond to increasing values of the truncation cutoff $E_\text{cut}$ and compute the corresponding spectra. 
If the effect of the perturbation weakens as the cutoff increases and it does so sufficiently fast, a low energy part of the sequence of spectra will converge in the sense that differences between successive spectra decrease with the cutoff. 
An approximation of the target model's spectrum can be obtained by extrapolation to infinite cutoff. 

Even though this approach sounds simple, its validity and efficiency relies on the choice of splitting in Eq.~\eqref{HT}. 
Splitting the Hamiltonian of a given QFT into the obvious free and interaction part does not always satisfy the condition of weakening interaction at increasing energy cutoffs. 
In this case the resulting spectra keep drifting as the cutoff increases. 
Moreover, there are frequently more than one equivalent forms of a QFT model related to each other by an exact nonlinear field transformation (duality) and the corresponding free and interaction part are completely different pieces of the Hamiltonian. 
Convergence may be satisfied in one splitting but not the other. 
Even when convergence is satisfied, it may be convenient to use a linear field transformation to recast the free Hamiltonian into an equivalent form that satisfies the same criterion but results in faster convergence. 

These issues can be partially addressed by means of \emph{renormalisation group} (RG) theory. 
RG theory shows that in the \emph{high energy limit} (UV limit) a general QFT model is described by an RG fixed point corresponding to a \emph{conformal field theory} (CFT). 
CFTs are models that are invariant under conformal transformations, a property that makes them exactly solvable. 
Perturbing a CFT by some operator results in a new model that may or may not have the same UV fixed point. 
This depends on the RG `scaling dimension' $\Delta$ of the perturbing operator in the given CFT. 
If $\Delta<1$ the perturbing operator is `relevant', if $\Delta>1$ it is `irrelevant', and if $\Delta=1$ it is `marginal'. 
Relevant perturbations modify the low energy spectrum of the original CFT model but not its high energy spectrum, therefore the resulting model has the same UV fixed point. 
In conclusion, a general QFT model can be constructed as a perturbation of the CFT that corresponds to its UV fixed point by an RG relevant operator. 
Based on the above, such operators would satisfy the condition of weakening of the perturbation at increasing cutoff. 
This guarantees convergence of HT for all RG relevant perturbations of CFT models~\cite{YUROV1990}. 
In fact, the smaller the scaling dimension the faster the convergence of the spectra. 

For the sG model RG theory predicts that the cosine interaction is an operator of scaling dimension $\Delta=\beta^2/(8\pi)$ therefore it is a relevant perturbation of the free massless boson for any $\beta<\sqrt{8\pi}$~\cite{sG_RG}. 
For the mS model in its original form the convergence criterion is not satisfied by the electromagnetic interaction between the fermionic fields. 
However, it is satisfied when we express the model in its bosonised form. 
The fermion interaction mediated by the electromagnetic field is mapped to the boson mass term of Eq.~\eqref{H_mS}, whereas the fermion mass term is mapped to the bosonic cosine interaction. 
In the UV limit the free massive boson reduces to the same CFT, therefore the cosine interaction is relevant with scaling dimension $\Delta=1/2$. 

Therefore, numerical methods for the study of QFT models can be based on truncation of the Hilbert space following an RG-prescribed splitting of the Hamiltonian into a free or, more generally, exactly solvable part, and the complementary part. 
Nevertheless, convergence of the numerically computed spectra is generally very slow (power-law) in most interesting cases, which combined with the exponential growth of the truncated Hilbert space severely limits the accuracy and range of validity of HT. 
Hence a more efficient TN-based implementation would be beneficial. 

\subsection*{Methods of \TNQF}

Following the ideas of HT, we use the free part of the target model's Hamiltonian to construct the Hilbert space basis and expand the interaction part in it. 
The free Hamiltonian of any QFT that is invariant under space translations can be diagonalised in Fourier space, where it decomposes into a system of independent harmonic oscillator modes (with the possible exception of the zero momentum mode). 
For a QFT defined in finite space with periodic boundary conditions, the corresponding Fourier decomposition is infinite but discrete, labelled by the mode numbers $k$ that take all integer values. 
The Hilbert space is therefore the tensor product of those of the momentum modes, each of which is spanned by states with any possible occupation number. 

Unlike conventional TNs that resemble the underlying lattice structure of the system in real space, our approach works entirely in Fourier space. 
In this respect, each individual MPS tensor represents one momentum mode $k$. 
A general state vector can be defined through its expansion in basis vectors $\ket{\{n_{k}\}}$ characterised by the mode occupation numbers, i.e., 
\begin{align}
    \ket{\Psi} = \sum_{\lbrace n_k \rbrace}
    \mathcal{S}_{\{n_{k}\}}
    \ket{\{n_{k}\}}, 
    \label{Psi}
\end{align}
where the sums run over all possible occupation numbers $n_k\in \mathbb N$ of all modes $k \in \mathbb Z$. 
In MPS theory, such a state vector can be approximated by a state of the form
\begin{align}
    \ket{\Psi_\mathrm{MPS}} = \sum_{\{n_k\}} \left( \sum_{\{\alpha_k\}}  \prod_{k} \, \lbrack \mathcal{M}^{(k)} \rbrack^{\alpha_{k+1}}_{\alpha_{k}, n_k} \right) \ket{\{n_k\}}
    \label{Psi_MPS2},
\end{align}
where $\mathcal{M}^{(k)}$ are site-dependent three-index tensors with $n_k$ being the `physical' indices corresponding to the mode occupation numbers, and $\alpha_k$ the `virtual' MPS indices. 
To obtain an approximation of the probability amplitudes $\mathcal{S}$ we contract the $\mathcal{M}$ tensors successively over all virtual indices. 
The MPS decomposition reduces the exponential scaling of the number of coefficients for $\mathcal S$ to a polynomial scaling for the product of $\mathcal{M}^{(k)}$, as long as the individual virtual indices $\alpha_k$ do not scale exponentially themselves. 
This is a typical property of low-energy states satisfying the area-law of entanglement entropy~\cite{Area_laws,CiracApproximability2008}, in which case MPS provide an efficient and accurate ansatz for the simulation of those quantum many-body system. 
The MPS representation in Eq.~\eqref{Psi_MPS2} is visualised in Fig.~\ref{fig:TN}. 

In order to perform a numerical computation we need to truncate the a priori infinite Hilbert space to a truncated space $\mathcal{H}_T$. 
This is achieved in two steps. 
First, we restrict the maximum absolute momentum which corresponds to a sharp UV cutoff and keep only the modes with $|k|\leq k_\mathrm{max}$. 
Second, we restrict the maximum occupation number of each mode keeping only the states with occupancy \mbox{$n_{k}\leq n(k)$}. 
Hence, the momentum modes in Eqs.~\eqref{Psi} and \eqref{Psi_MPS2} are restricted to $k \in \{-k_\mathrm{max}, \dots, +k_\mathrm{max}\}$ instead of $\mathbb{Z}$, and similarly $n_k \in \{0, \dots, n(k)\}$ instead of $\mathbb{N}$. 
The virtual indices $\alpha_k$ take values in auxiliary spaces whose dimensions $\chi_k$, the `bond dimensions', are automatically adjusted in the course of the computation according to a chosen error tolerance to capture the most relevant quantum correlations in the system. 
Hence, in the numerical implementation, the length of the MPS is $2k_\mathrm{max}+1$ and the variational parameters are contained in the individual tensors $\mathcal M^{(k)}$. 
The efficient MPS representation enables us to reach (truncated) Hilbert spaces of dimension of the order of $\operatorname{dim}(\mathcal H_\text{T}) \sim 10^{10}$, 
that are inaccessible with standard HT methods based on exact diagonalization. 

The occupation number cutoffs $n(k)$ fully determine the truncation applied and can be adjusted and optimised dynamically in the course of the computation according to a global error tolerance condition. 
However, based on extensive efficiency and accuracy tests, we find that the following simple choice is suitable for most purposes. 
Given that the modes above the UV cutoff are omitted on the basis of the weakening of the interaction for increasing momentum, it is reasonable to choose $n(k)$ as a decreasing function of $|k|$, specifically the integer part of $n_\mathrm{max}/|k|$ where $n_\mathrm{max}$ is the maximum occupation of the first mode (and maximum occupation in all modes) and it is an additional truncation parameter that should be $n_\mathrm{max}\geq k_\mathrm{max}$. 
A common choice is setting $n_\mathrm{max}=k_\mathrm{max}$, so that $n(k)$ decreases to $n(\pm k_\mathrm{max}) = 1$ at the edges of the MPS. 
However, we may choose $n(k)$ in different ways (e.g., for the purposes of UV cutoff extrapolation, as explained in the Supplementary Material~\cite{SM}). 
The zero mode cutoff is chosen independently according to an error tolerance test. 

Having defined the truncated subspace and its MPS variational representation, we should now construct the sG and mS Hamiltonians as MPOs acting on this subspace. 
For both models the free part Hamiltonian is a simple sum of single mode terms which can be easily represented in MPO form. 
The cosine interaction potential is the integral over all space of a sum of two imaginary exponentials of the field $\Phi$. 
Constructing such exponential operators in TN form is straightforward, too. 
Using the expansion of $\Phi$ in momentum modes, they are expressed as products of single mode operators $\mathcal V^{(k)}$ acting separately on each of them (see Supplementary Material for more details). 
However, performing the spatial integration of the interaction potential is more challenging. 
As explained in the main text, this last step of the MPO construction is equivalent to imposing momentum conservation as a global constraint. 
A single exponential operator is then constructed according to
\begin{align}
    \hat O &= 
    \sum_{\lbrace n_k \rbrace} \sum_{\lbrace n_k^\prime \rbrace} \left\lbrack \sum_{\lbrace \delta_k \rbrace}  \left( \prod_{k} \, \lbrack \mathcal V^{(k)}\rbrack_{n_k}^{n_k^\prime,\delta_k} \right)  \mathsf{D}(\lbrace \delta_k \rbrace) \right\rbrack \nonumber \\
    & \qquad \times \ket{\{n_{k}\}} \bra{\{n_{k}^\prime\}} .
\end{align}
Here, $n_k$ ($n_k^\prime$) are the individual occupation numbers of the ket (bra) indices of the single mode operators $\mathcal V^{(k)}$, and \mbox{$\delta_k = k(n_k - n_k^\prime)$} denotes the momentum transfer on each site. 
The global constraint $\mathsf{D}(\lbrace \delta_k \rbrace)$ is a central object in the final MPO construction, that can be \emph{efficiently and exactly} implemented using \emph{symmetric tensor networks} (STN)~\cite{STN_0, STN}. 

It should be noted that there are some important differences between the present problem and typical applications of STN. 
First, in local lattice models that satisfy a global symmetry (e.g., total particle number conservation), the MPO representation of the Hamiltonian constructed by adding local terms is already symmetric and using STN is merely an option for efficiency improvement. 
In contrast, here imposing the symmetry constraint is a necessary separate step in the construction of the Hamiltonian MPO. 
Second, given that the interaction is non-local in momentum space but couples all momentum modes to all others, symmetry sector projectors represented as symmetric MPS are not expected to have a low bond dimension (we will comment on this below). 

In the present case the construction of a projector on a sector with total momentum $P$ corresponds to finding the space of solutions of the equation $ P = (2\pi/L) \sum_{k} k n_k $. 
This belongs to the class of `counting problems', a special class of \emph{constraint satisfaction problems} that are generally much harder than `decision problems'. 
The global constraint $\mathsf{D}(\lbrace \delta_k \rbrace)$ is precisely a projector onto a total momentum $P$ sector, that falls into that class. 
TNs can provide efficient construction methods for this type of problems~\cite{Liu2023}. 
In our method an exact TN representation of arbitrary projectors to $P$-sectors is constructed employing STN. 
This is achieved by constructing an MPS for an equal-weight superposition of all possible combinations of mode occupations that satisfy the global constraint and contracting it with the MPO that should be projected. 
The final exponential MPO used to generate the cosine interaction is, therefore, given by
\begin{align}
    & \hat O_\mathrm{MPO} = \sum_{\lbrace n_k \rbrace} \sum_{\lbrace n_k^\prime \rbrace} 
    \left( \sum_{\{\alpha_k\}} \prod_{k} \, \lbrack \mathcal T^{(k)} \rbrack_{\alpha_k,n_k}^{n_k^\prime,\alpha_{k+1}} \right) \ket{\{n_{k}\}} \bra{\{n_{k}^\prime\}} \nonumber \\ 
    & \text{ with } \quad \lbrack \mathcal T^{(k)} \rbrack_{\alpha_k,n_k}^{n_k^\prime,\alpha_{k+1}} := \sum_{\{\delta_k\}} \lbrack \mathcal V^{(k)}\rbrack_{n_k}^{n_k^\prime,\delta_k} \lbrack \mathcal D^{(k)} \rbrack_{\alpha_k, \delta_k}^{\alpha_{k+1}}.
\end{align}
Here, $\mathcal D^{(k)}$ denotes the three-index MPS tensors of the projector, and contracting over each of the momentum transfer indices $\delta_k$ yields a regular four-index MPO tensor $\mathcal T^{(k)}$. 
The dimension of the virtual indices $\alpha_k$ of the MPS representation of $\mathsf{D}(\lbrace \delta_k \rbrace)$ depend on $k_\text{max}$ and $n(k)$ of the truncated Hilbert space. 
The general construction of the interacting MPO is visualized in Fig.~\ref{fig:TN} and explained in detail in the Supplementary Material. 

This step plays a crucial role in our technique and is the one that mainly influences its efficiency, as the bond dimension of the projection MPS is much larger than those of the other constituents of the resulting MPO representing the Hamiltonian. 
Specifically, we find that its bond dimension increases polynomially with the total mode number. 
A correspondingly large MPS bond dimension $\chi$ is required with increasing MPO bond dimension, such that the number of modes we can include in the simulations is typically limited by run-time and memory capacity of a typical HPC cluster. 

Once the MPO representation of the Hamiltonian is obtained, the simulation of the model is done using standard TN techniques. 
We use two-site DMRG for the computation of the ground state~\cite{Schollwoeck2011} and the first excited state~\cite{Stoudenmire2012}, as well as two-site TDVP~\cite{Haegeman2011} for the simulation of quench dynamics. 
We found that a global subspace expansion~\cite{Yang2020} is necessary to obtain correct results for quench dynamics due to the non-local (in momentum space) character of the interactions and the non-uniform truncated Hilbert spaces of the momentum modes. 
Generally, the dominant truncation parameters affecting the MPS bond dimension are set to a maximal error of $\epsilon_D = 10^{-6}$ and $\epsilon_T = 10^{-4}$ for the two-site DMRG and TDVP algorithms, respectively, or a fixed limit of $\chi = 2500$. 
Details are presented in the Supplementary Material.

\begin{acknowledgments}

We want to thank David Horváth and Gábor Takács, as well as Mari Carmen Bañuls, Andreas Haller, Shozab Qasim, Matteo Rizzi and Frank Verstraete for inspiring discussions. This work has been funded by the Deutsche Forschungsgemeinschaft (DFG, German Research Foundation) under the project number 277101999 -- CRC 183 (project B01), and the BMBF (MUNIQC-Atoms, FermiQP). J.~E.~acknowledges funding of the ERC (DebuQC). S.~S. acknowledges support by the European Union’s Horizon 2020 research and innovation programme under the Marie Sk\l{}odowska-Curie grant agreement No.~101030988. We also acknowledge the use of \mbox{TensorKit}~\cite{TensorKitGitHub} for the numerical simulations.

\end{acknowledgments}

\bibliography{HTTN_references}

\newpage

\appendix

\section*{SUPPLEMENTARY INFORMATION}

\section{Model definition}
\label{app:modelDefinition}

\subsection{Introduction}

In this Supplementary Information, we explain in detail the construction of the \emph{sine-Gordon} (sG) and \emph{massive Schwinger} (mS) Hamiltonians in the Fock bases of the free massless and massive boson Hamiltonians, respectively. 
We first focus on the sG model and treat the mS model afterwards. 

\subsection{Sine-Gordon model}
\label{app:SG}

As discussed in the main text, the sG Hamiltonian can be written in the form 
\begin{equation}
    H_{\mathrm{sG}} = H_{0} + \lambda V
    \label{H_SG},
\end{equation}
where $H_{0}$ is the free massless boson Hamiltonian and $V$ the cosine interaction
\begin{align}
    H_{0} & =\frac{1}{2}\int_{0}^{L}\mathrm{d}x\,\left(\Pi^{2}+(\partial_{x}\Phi)^{2}\right),\\
    V & = - \int_{0}^{L}\mathrm{d}x \, \twoOrd{\cos\beta\Phi}. 
\end{align}
Here, $\beta$ is the frequency of the cosine interaction, $L$ is the system size, and we use units $\hbar=1, c=1$. 
The normal ordering of the interaction, denoted by $\twoOrd{\cdot}$, is performed in reference to the mode expansion of the free Hamiltonian and will be precisely defined in the following.

The quantum fields $\Phi$ and $\Pi$ satisfy the canonical commutation relations 
\begin{equation}
    [\Phi(x),\Pi(x')]=\mathrm{i}\delta(x-x').
\end{equation}
The field $\Phi$ is `compactified', i.e., it is an angular field taking values on a circle of radius
\begin{equation}
    R = 1/\beta.
\end{equation}
We assume periodic boundary conditions, which taking into account the field compactification, read
\begin{equation}
    \Phi(L)=\Phi(0)+2\pi Rm, \text{ for any } m \in \mathbb{Z}.
\end{equation}
We will use the Fock basis of $H_{0}$ as the basis on which we will perform the truncation and construct the full Hamiltonian and observables. 
To this end, we need to expand the two conjugate fields in the diagonalisation modes of $H_{0}$. 

\subsubsection{Mode expansion -- Fock basis}
\label{subsub:modeExpansionFockBasis}

The expansion of the field $\Phi(x)$ in the modes that diagonalise $H_{0}$ for the given periodic boundary conditions is 
\begin{equation}
    \Phi(x)=\tilde{\Phi}_{0}+2\pi RM\frac{x}{L}+\sum_{k \in \mathbb{Z}^*} f_{k}(x) \tilde{\Phi}_{k},
    \label{eq:mode_expansion_phi}
\end{equation}
where the eigenfunctions $f_{k}(x)$ are plane waves
\begin{equation}
    f_{k}(x)=\frac{1}{\sqrt{L}}\mathrm{e}^{\mathrm{i}p_{k}x}
\end{equation}
with momenta $p_{k}$ quantised according to
\begin{equation}
    p_{k}=2\pi k/L,\;k\in\mathbb{Z}^{*}.
\end{equation}
Here, we use $\mathbb Z^* = \mathbb Z \setminus \lbrace 0\rbrace$ to denote integer numbers excluding zero. 
The operator $M$ is the winding number operator, which has eigenvalues $m\in\mathbb{Z}$ counting the number of times the continuous field $\Phi$ winds around the circle when moving from the left to the right edge, $M=\left(\Phi(L)-\Phi(0)\right)/(2\pi R)$. 
The sum of all other terms in Eq.~\eqref{eq:mode_expansion_phi} is the usual Fourier series of periodic functions. 

The expansion of the canonical momentum field $\Pi$ is given by
\begin{equation}
    \Pi(x)=\tilde{\Pi}_{0}+\sum_{k \in \mathbb Z^*}f_{k}(x)\tilde{\Pi}_{k}\label{eq:mode_expansion_pi},
\end{equation}
with canonical commutation relations 
\begin{align}
    \left\lbrack \tilde{\Phi}_{0},\tilde{\Pi}_{0} \right\rbrack & = \mathrm{i},\label{eq:CCR0}\\{}
    \left\lbrack \tilde{\Phi}_{k},\tilde{\Pi}_{k'}\right\rbrack & = \mathrm{i}\delta_{k,-k'}.\label{eq:CCR}
\end{align}
In terms of these modes, the free Hamiltonian is transformed into
\begin{equation}
    \begin{split}
    H_{0} &= \frac{1}{2L} \tilde{\Pi}_{0}^{2}+\frac{(2\pi R)^{2}}{2L}M^{2} \\
    & +\frac{1}{2}\sum_{k \in \mathbb Z^*}\left(\tilde{\Pi}_{k}\tilde{\Pi}_{-k}+p_{k}^{2}\tilde{\Phi}_{k}\tilde{\Phi}_{-k}\right),
    \label{eq:H0_mode_exp}
    \end{split}
\end{equation}
as a simple calculation shows.

\subsubsection*{Non-zero modes}

All modes with $k \neq 0$, which correspond to non-zero momenta, appear in the Hamiltonian as decoupled harmonic oscillators with frequencies
\begin{equation}
    \omega_{k}=|p_{k}|,\,k\in\mathbb{Z}^{*}.
\end{equation}
Therefore, in order to diagonalise the Hamiltonian $H_{0}$ we define ladder operators for each of these modes through the relations
\begin{align}
    \tilde{\Phi}_{k} & =\frac{1}{\sqrt{2\omega_{k}}}\left(A_{k}+A_{-k}^{\dagger}\right) , \\
    \tilde{\Pi}_{k} & =(-{\rm i})\sqrt{\frac{\omega_{k}}{2}}\left(A_{k}-A_{-k}^{\dagger}\right),
    \label{eq:ladder_ops}
\end{align}
with the standard canonical commutation relations
\begin{align}
    \left\lbrack A_{k},A_{k'}^{\dagger} \right\rbrack & =\delta_{k,k'}.
\end{align}
Overall, the mode expansion of $\Phi
$ in terms of ladder operators can be expressed as 
\begin{equation}
    \begin{split}
        \Phi(x) & = \tilde{\Phi}_{0}+2\pi RM\frac{x}{L} \\
        & +\frac{1}{\sqrt{L}}\sum_{k \in \mathbb{Z}^*}\frac{1}{\sqrt{2\omega_{k}}}\left(A_{k}\mathrm{e}^{+\mathrm{i}p_{k}x}+A_{k}^{\dagger}\mathrm{e}^{-\mathrm{i}p_{k}x}\right).
    \label{eq:mode_expansion_phi_2}
    \end{split}
\end{equation}

\subsubsection*{Zero modes}

Unlike for the non-zero modes, the operators $\tilde{\Pi}_{0}$ and $M$ appear in the free Hamiltonian unaccompanied by their respective canonical conjugates, so that they are non-harmonic modes. 
The free Hamiltonian is already diagonal in their eigenspaces. 
Given that the boson field $\Phi$ takes values on a circle of radius $R$, i.e., in field theory terminology, it is `compactified' with `compactification radius' $R$, field values that differ by integer multiples of $2\pi R$ are considered identical. 
As we have already seen, this means that, imposing periodic boundary conditions, winding field configurations are allowed, i.e., after unwrapping a continuous field configuration $\Phi$ the difference of its values at the edges $x=0$ and $x=L$ can be any integer multiple of $2\pi R$ 
\begin{equation}
    \Phi(L) = \Phi(0)+2\pi Rm,\quad m\in\mathbb{Z}.
\end{equation}
The field $\Phi$ is periodically continuous at the boundaries $x=0$ and $x=L$ but modulo $2\pi R$. 
The integer $m$ is called the `winding number' and it expresses how many times the field winds around the circle in the positive direction when we move from one edge to the other. 

However, the compactified nature of the quantum field has an additional important consequence. 
Conceptually, only observables constructed out of space or time derivatives of $\Phi(x)$ and imaginary exponentials of the form $\mathrm{e}^{\mathrm{i}n\Phi(x)/R}$ for integer $n$ are acceptable, and the Hamiltonian of a physical model can only involve such operators. 
Note, in particular, that the operator $\tilde{\Phi}_{0}$ is not well-defined as it is not single-valued. 

Next, note that from the zero mode canonical commutation relation in Eq.~\eqref{eq:CCR0} we have
\begin{equation}
    \mathrm{e}^{-\mathrm{i}n\tilde{\Phi}_{0}/R}\tilde{\Pi}_{0}\mathrm{e}^{+\mathrm{i}n\tilde{\Phi}_{0}/R}=\tilde{\Pi}_{0}+\frac{n}{R},
    \label{eq:zero-mode_exp_op}
\end{equation}
which means that the exponential operators $\mathrm{e}^{\mathrm{i}n\tilde{\Phi}_{0}/R}$ shift the eigenvalues of $\tilde{\Pi}_{0}$. 
Indeed, if $|\psi\rangle$ is an eigenstate of $\tilde{\Pi}_{0}$ with eigenvalue $\pi_{0}$, then the state vector $\mathrm{e}^{\mathrm{i}n\tilde{\Phi}_{0}/R}|\psi\rangle$ is another eigenstate of $\tilde{\Pi}_{0}$ with eigenvalue $\pi_{0}+n/R$. 
Since these exponential operators are well-defined only for integer $n$, the spectrum of $\tilde{\Pi}_{0}$ is restricted to $\mathbb{Z}/R$, and we can label its eigenstates with an integer $\ell$ so that 
\begin{equation}
    \tilde{\Pi}_{0}|\ell\rangle = \frac{\ell}{R} |\ell\rangle, \quad \ell\in\mathbb{Z}. 
    \label{eq:zero-mode_eigstates_1}
\end{equation}
From Eq.~\eqref{eq:zero-mode_exp_op} we then find
\begin{equation}
    \mathrm{e}^{\mathrm{i}n\tilde{\Phi}_{0}/R}|\ell\rangle=|\ell+n\rangle . 
    \label{eq:zero-mode_eigstates_2}
\end{equation}
An important consequence of this is that the spectrum of the first Hamiltonian term in Eq.~\eqref{eq:H0_mode_exp}, which is proportional to $\tilde{\Pi}_{0}^{2}$, is also discrete like those of the second term, which is proportional to $M^{2}$, and of the non-zero modes. 

\subsubsection{Hilbert space}

The Hilbert space of the model is the tensor product of the eigenspaces of $\tilde{\Pi}_{0}$, of $M$ and of the Fock spaces of the infinite set of harmonic non-zero momentum modes. 
The basis state vectors are characterised by the eigenvalues of $\tilde{\Pi}_{0}$, $M$ and the sequence $\{n_{k}\}$ of occupation numbers of the non-zero modes. 
Concretely, they are given by
\begin{equation}
    |\Psi\rangle=|\ell,m;\{n_{k}\}\rangle=\prod_{k \in \mathbb Z^*}\frac{A_{k}^{\dagger n_{k}}}{\sqrt{n_{k}!}}|\ell,m;0\rangle
    \label{eq:basis_states}
\end{equation}
for $\ell\in\mathbb{Z}$, $m\in\mathbb{Z}$ and $n_{k}\in\mathbb{N}$.

\subsubsection{Free part of the Hamiltonian}
\label{sec:SG_H0}

In terms of the mode operators defined above, the free Hamiltonian is
\begin{align*}
    H_{0} & = \frac{1}{2}L\tilde{\Pi}_{0}^{2}+\frac{(2\pi R)^{2}}{2L}M^{2}+\frac{1}{2}\sum_{k \in \mathbb Z^*}\omega_{k}\left(A_{k}A_{k}^{\dagger}+A_{k}^{\dagger}A_{k}\right) \\
    & = \frac{1}{2}L\tilde{\Pi}_{0}^{2}+\frac{(2\pi R)^{2}}{2L}M^{2}+\sum_{k \in \mathbb Z^*}\omega_{k}A_{k}^{\dagger}A_{k}+\frac{1}{2}\sum_{k \in \mathbb Z^*}\omega_{k}.
\end{align*}
The last term is an infinite constant corresponding to the vacuum state energy of the infinite set of harmonic modes. 
Normal ordering removes this infinite additive constant and we obtain
\begin{align}
    \twoOrd{H_{0}} = \frac{1}{2}L\tilde{\Pi}_{0}^{2}+\frac{(2\pi R)^{2}}{2L}M^{2}+\sum_{k \in \mathbb Z^*}\omega_{k} A_{k}^{\dagger}A_{k}.
    \label{eq:normalOrderedFreeHamiltonianSG}
\end{align}
As already noted, the zero mode terms (first two in Eq.~\eqref{eq:normalOrderedFreeHamiltonianSG}) correspond to projections to the eigenspaces of the operators $\tilde{\Pi}_{0}$ and $M$. 
The eigenvalues of $\twoOrd{H_{0}}$ in the basis state vectors $|\Psi\rangle$ are
\begin{equation}
    \begin{split}
        E[\Psi] &= \langle\Psi| \twoOrd{H_{0}} |\Psi\rangle \\
        & = \frac{1}{2}L\frac{\ell^{2}}{R^{2}}+\frac{(2\pi R)^{2}}{2L}m^{2}+\sum_{k  \in \mathbb Z^*} n_{k} \omega_{k}.
    \label{H_0_me}
    \end{split}
\end{equation}

\subsubsection{Interacting part of the Hamiltonian}

\subsubsection*{Vertex operator}

Having defined the Hilbert space basis and constructed $H_{0}$, our next step is to calculate the matrix elements of $V$ in this basis. 
We first define exponential operators of the form 
\begin{align}
    V_{\alpha}(x) = \mathrm{e}^{\mathrm{i}\alpha\Phi(x)}.
    \label{eq:vertexOperatorDefinition}
\end{align}
Based on the earlier discussion on consequences of compactification, exponential operators of the above form are well-defined (single-valued) only if $\alpha$ is an integer multiple of $1/R$, i.e. 
\begin{equation}
    \alpha = \frac{n}{R}, \quad n \in \mathbb{Z}.
    \label{eq:alpha_condition}
\end{equation}
Using the mode expansion in Eq.~\eqref{eq:mode_expansion_phi_2} and given
that different modes commute, we can write $V_{\alpha}(x)$ as a product of single mode exponentials. 
Moreover, from Eq.~\eqref{eq:ladder_ops} and using the Baker–Campbell–Hausdorff formula 
\begin{equation}
    \mathrm{e}^{A+B}=\mathrm{e}^{A}\mathrm{e}^{B}\mathrm{e}^{-\tfrac{1}{2}[A,B]},
\end{equation}
which is valid when $[A,B]$ commutes with $A$ and $B$, we can rewrite $V_{\alpha}(x)$ as a normal-ordered exponential, where each non-zero mode factor is cast in the form 
\begin{align}
    & \mathrm{e}^{\frac{\mathrm{i}\alpha}{\sqrt{2\omega_{k}}}\left(f_{k}(x)A_{k}+f_{k}^{*}(x)A_{k}^{\dagger}\right)} \\
    \nonumber
    & \qquad = \mathrm{e}^{-\tfrac{1}{2}\alpha^{2}\frac{|f_{k}(x)|^{2}}{2\omega_{k}}}\mathrm{e}^{\mathrm{i}\alpha\frac{f_{k}^{*}(x)}{\sqrt{2\omega_{k}}}A_{k}^{\dagger}}\mathrm{e}^{\mathrm{i}\alpha\frac{f_{k}(x)}{\sqrt{2\omega_{k}}}A_{k}} 
    \nonumber\\
    & \qquad = \mathrm{e}^{-\tfrac{1}{2}\alpha^{2}\frac{|f_{k}(x)|^{2}}{2\omega_{k}}} \twoOrd{\mathrm{e}^{\frac{\mathrm{i}\alpha}{\sqrt{2\omega_{k}}}\left(f_{k}(x)A_{k}+f_{k}^{*}(x)A_{k}^{\dagger}\right)}}.
    \nonumber
\end{align}
Therefore, the relation between $V_{\alpha}(x)$ and its normal-ordered form $\twoOrd{V_{\alpha}(x)}$ is 
\begin{equation}
    V_{\alpha}(x) = \mathrm{e}^{-\tfrac{1}{2}\alpha^{2} \sum_{k \in \mathbb{Z}^*} \frac{|f_{k}(x)|^{2}}{2\omega_{k}}} \twoOrd{V_{\alpha}(x)}.
\end{equation}
For periodic boundary conditions, the series in the exponent of the proportionality factor is given by
\begin{align}
    \sum_{k \in \mathbb Z^*}\frac{|f_{k}(x)|^{2}}{2\omega_{k}} & =\frac{1}{4\pi}\sum_{k \in \mathbb Z^*}\frac{1}{|k|}=\frac{1}{2\pi}\sum_{k=1}^{\infty}\frac{1}{k},
\end{align}
which diverges. 
This is once again due to the infinite number of harmonic modes. 
It is, therefore, more meaningful to use the normal-ordered exponential operators $\twoOrd{V_{\alpha}(x)}$, and from now on we focus on these. 

Their matrix elements between arbitrary pairs of basis states of Eq.~\eqref{eq:basis_states} are 
\begin{align*}
    \langle\Psi'| \twoOrd{V_{\alpha}(x)} |\Psi\rangle & = \langle \ell', m'; \{n'_{k}\}| \twoOrd{\mathrm{e}^{\mathrm{i}\alpha\Phi(x)}} |\ell, m; \{n_{k}\} \rangle\\
    & = \langle\ell'| \mathrm{e}^{\mathrm{i}\alpha\tilde{\Phi}_{0}} |\ell\rangle\,\langle m'|\mathrm{e}^{\mathrm{i}2\pi\alpha RM\frac{x}{L}}|m\rangle \\
    & \times \prod_{k \in \mathbb Z^*} \langle n'_{k}|\mathrm{e}^{\mathrm{i}\alpha f_{k}^{*}(x)\frac{A_{k}^{\dagger}}{\sqrt{2\omega_{k}}}}\,\mathrm{e}^{\mathrm{i}\alpha f_{k}(x)\frac{A_{k}}{\sqrt{2\omega_{k}}}}|n_{k}\rangle.
\end{align*}
The first factor can be calculated using Eq.~\eqref{eq:zero-mode_eigstates_2}. 
Together with Eq.~\eqref{eq:alpha_condition} we then find 
\begin{align}
    \langle\ell'| \mathrm{e}^{\mathrm{i}\alpha\tilde{\Phi}_{0}} |\ell\rangle = \delta_{\ell',\ell+\alpha R}, \quad (\text{for }\alpha\in\mathbb{Z}/R).
    \label{eq:vertexOperatorZeroMode_sG}
\end{align}
The second factor is trivially given by
\begin{equation}
    \langle m'| \mathrm{e}^{\mathrm{i}2\pi\alpha RM\frac{x}{L}} |m\rangle = \mathrm{e}^{\mathrm{i}2\pi\alpha Rm\frac{x}{L}}\delta_{m',m}.
\end{equation}
To calculate the non-zero mode factors, we first define the function
\begin{equation}
    F(n',n;z',z)=\langle n'|\mathrm{e}^{z'A^{\dagger}}\,\mathrm{e}^{zA}|n\rangle
\end{equation}
for a single harmonic oscillator. 
Its matrix elements can be found analytically as 
\begin{align}
    & F(n',n;z',z) \\
    & =\langle0|\frac{A^{n'}}{\sqrt{n'!}}\mathrm{e}^{z'A^{\dagger}}\,\mathrm{e}^{zA}\frac{A^{\dagger n}}{\sqrt{n!}}|0\rangle\nonumber \\
    & =\sum_{m,m'=0}^{\infty}\frac{z'^{m'}z^{m}}{m'!m!}\frac{1}{\sqrt{n'!n!}}\langle0|A^{n'}A^{\dagger m'}A^{m}A^{\dagger n}|0\rangle\nonumber \\
    & =\sqrt{n'!n!}\sum_{m,m'=0}^{\infty}\frac{z'^{m'}z^{m}}{m'!m!\left(n-m\right)!}\delta_{n'-m',n-m}\Theta(n-m) \nonumber \\
    & =\sqrt{n'!n!}\sum_{m=\max(0,n-n')}^{n}\frac{z'^{m+n'-n}z^{m}}{m!(n'-n+m)!(n-m)!},\nonumber 
\end{align}
where we have used that 
\begin{align}
    \langle0|A^{n'}A^{\dagger m'}A^{m}A^{\dagger n}|0\rangle & = \begin{cases}
        \frac{n'!n!}{\left(n-m\right)!} & \parbox{2.8cm}{if $n'-m'=n-m$ and $n\geq m$,}\\
        0 & \text{ otherwise.}
    \end{cases}
\end{align}
In our special case, the argument is \mbox{$z = \mathrm{i} \alpha f_{k}(x) / \sqrt{2\omega_{k}}$} and we have \mbox{$z'=-z^{*}$}, hence
\begin{align}
    \begin{split}
        & F(n',n;z,-z^{*}) \\
        & =\begin{cases}
            \sqrt{\frac{n!}{n'!}}z^{n'-n}L_{n}^{(n'-n)}\left(|z|^{2}\right) & \text{for }n'\geq n,\\
            \sqrt{\frac{n'!}{n!}}\left(-z^{*}\right)^{n-n'}L_{n'}^{(n-n')}\left(|z|^{2}\right) & \text{ for }n'<n.
        \end{cases}
    \end{split}
    \label{eq:F}
\end{align}
Here $x\mapsto L_{n}^{(\ell)}\left(x\right)$ are the associated Laguerre polynomials. 
This formula is also found in the derivation of the Fock-basis representation of the displacement operator in the context of coherent states of the quantum harmonic oscillator. 
Replacing \mbox{$z = \mathrm{i}\alpha f_{k}(x)/\sqrt{2\omega_{k}}$} we can express Eq.~\eqref{eq:F} as
\begin{align}
    F(n',n;z,-z^{*}) &= \mathrm{e}^{\mathrm{i}(n'-n)p_{k}x} G_{n',n} \left( \frac{\alpha^{2}}{2\omega_{k}L} \right),
    \label{eq:F2}
\end{align}
with analytic functions
\begin{align}
    \begin{split}
    & G_{n',n}\left(\frac{\alpha^{2}}{2\omega_{k}L}\right) \\ 
    & = \begin{cases}
        \sqrt{\frac{n!}{n'!}}L_{n}^{(n'-n)}\left(\frac{\alpha^{2}}{2\omega_{k}L}\right) & \text{for }n'\geq n,\\
        \sqrt{\frac{n'!}{n!}}\left(-1\right)^{n-n'}L_{n'}^{(n-n')}\left(\frac{\alpha^{2}}{2\omega_{k}L}\right) & \text{ for }n'<n.
        \end{cases}
    \end{split}
    \label{eq:G}
\end{align}
Finally, putting everything together, we find
\begin{align}
    \begin{split}
        \langle\Psi'| \twoOrd{V_{\alpha}(x)} |\Psi\rangle &= \delta_{\ell',\ell+\alpha R} \, \delta_{m', m} \\ & \times \mathrm{e}^{\mathrm{i}\left[2\pi\alpha Rm/L+\sum_{k \in \mathbb Z^*}(n'_{k}-n_{k})p_{k}\right]x} \\
        & \times \prod_{k \in \mathbb Z^*}G_{n'_{k},n_{k}}\left(\frac{\alpha^{2}}{2\omega_{k}L}\right),
    \end{split}
    \label{eq:V_mat_elmnts}
\end{align}
where, given that the harmonic mode frequencies are 
\begin{equation}
    \omega_{k}=|p_{k}|=\frac{2\pi}{L}|k|,
\end{equation}
the argument of the function $G$ becomes \mbox{$\alpha^{2}/(2\omega_{k}L)=\alpha^{2}/(4\pi|k|)$}.

\subsubsection*{Space integrated vertex operator}
\label{sub:spaceIntegratedInteractionOperator}

Setting $\beta=1/R$, we construct the cosine interaction potential as the sum of two normal-ordered vertex operators 
\begin{align}
    \twoOrd{\cos\beta\Phi(x)} = \frac{1}{2} \left( \twoOrd{V_{+\beta}(x)} + \twoOrd{V_{-\beta}(x)} \right).
    \label{eq:normalOrderedCosineInteraction}
\end{align}
Having computed the matrix elements of vertex operators in the free Hamiltonian basis, the last step required for the construction of the interaction Hamiltonian is the integration of the vertex operator over all space, which gives 
\begin{align*}
    & \int_{0}^{L} \mathrm{d}x \; \langle\Psi'| \twoOrd{V_{\pm\beta}(x)} |\Psi\rangle \\
    & = \int_{0}^{L} \mathrm{d}x \; \mathrm{e}^{\mathrm{i}\left[\pm 2\pi m/L + \sum_{k}(n'_{k} -n_{k})p_{k}\right]x} \\ & \times \delta_{\ell',\ell\pm 1} \, \delta_{m'm} \, \prod_{k \in \mathbb Z^*} G_{n'_{k},n_{k}} \left( \frac{\beta^{2}}{2\omega_{k}L} \right) \\
    & = L \delta_{\pm m+\sum_{k}(n'_{k}-n_{k})k,0} \, \delta_{\ell',\ell\pm 1} \, 
    \delta_{m',m} \, \prod_{k \in \mathbb Z^*} G_{n'_{k},n_{k}} \left( \frac{\beta^{2}}{2\omega_{k}L} \right).
\end{align*}
For $m = 0$ this step imposes conservation of the total momentum. 
Indeed the total momentum operator is
\begin{equation}
    P = \sum_{k \in \mathbb Z^*} p_{k} A_{k}^{\dagger}A_{k} 
\end{equation}
and its eigenvalues in basis states of Eq.~\eqref{eq:basis_states}
\begin{equation}
    P[\Psi] = \sum_{k \in \mathbb Z^*} n_{k} p_{k} = \frac{2\pi}{L}\sum_{k \in \mathbb Z^*} n_{k} k.
\end{equation}
In fact, the full space integrated vertex operator is thus momentum-preserving, i.e., it does not change the momentum of the state vector $\ket{\Psi}$ it acts upon. 
This property will be a central part in the construction of the operator in tensor network language.

\subsubsection{sG Hamiltonian}

Putting everything together and restricting to the \mbox{$M = 0$} sector of the Hilbert space, the matrix elements of the sG Hamiltonian $H_\mathrm{sG} = H_{0} + \lambda V$ can finally be expressed as
\begin{align}
    \begin{split}
    \langle\Psi'|\twoOrd{H_0}|\Psi\rangle &= \frac{1}{2}L\frac{\ell^{2}}{R^{2}}+\frac{(2\pi R)^{2}}{2L}m^{2} \\
    & + \left(\frac{2\pi}{L}\right)\sum_{k \in \mathbb Z^*} n_{k} |k|,
    \end{split}\\
    \begin{split}
    \langle\Psi'| V |\Psi\rangle  
    & = - \frac{1}{2} L \delta_{P[\Psi'],P[\Psi]} \, \left(\delta_{\ell',\ell+1} + \delta_{\ell',\ell-1}\right) \\
    & \times \prod_{k \in \mathbb Z^*} G_{n'_{k},n_{k}} \left(\frac{\beta^{2}}{4\pi|k|}\right).
    \end{split}
\end{align}
It is customary to express the coupling constant $\lambda$ in Eq.~\eqref{H_SG} in terms of the physical parameters of the model, specifically, the mass $M_s$ of the soliton, the fundamental particle of the model out of which composite particles (the breathers) are made, and the cosine frequency $\beta$, reparametrised as 
\begin{equation}
    \Delta = \frac{\beta^2}{8\pi}. 
\end{equation}
As mentioned in the main text, $\Delta$ is the \emph{RG scaling dimension} of the exponential operator and determines if it is RG relevant, marginal or irrelevant. 
The soliton mass depends on $\beta$ due to the `dressing' effect of the interaction on particle mass, as explained by RG theory. 
In terms of the above, the coefficient $\lambda$ is set to~\cite{ZAMOLODCHIKOV1995}
\begin{equation}
    \lambda = M_s^{2} \left(\frac{2 \pi}{M_s L}\right)^{2 \Delta}  \kappa(\Delta),
\end{equation}
where 
\begin{equation}
    \kappa(\Delta) = 
    \frac{2\Gamma(\Delta)}{\pi \Gamma(1-\Delta)}\left[\frac{\sqrt{\pi} \Gamma\left(\frac{1}{2-2 \Delta}\right)}{2 \Gamma\left(\frac{\Delta}{2-2 \Delta}\right)}\right]^{2-2 \Delta}.
\end{equation}

\subsubsection{sG energy gap}

In QFT the energy gap, i.e., the difference between the ground and first excited state energies in the thermodynamic limit, equals the mass of the lightest particle of the model. 
For $\beta<\sqrt{4\pi}$ the sG model is attractive, therefore there are one or more bound states between soliton--antisoliton pairs, corresponding to breathers. 
In this regime the energy gap is the mass of the first breather, which is predicted to be~\cite{ZAMOLODCHIKOV1995} 
\begin{equation}
    \Delta_\mathrm{sG} = m_\mathrm{B1} = 2 M_s \sin\left(\frac{\pi \Delta/2}{1-\Delta}\right) \quad (\beta<\sqrt{4\pi}).
\end{equation}
This particle consists of one soliton and one antisoliton. 
At the free fermion point $\beta=\sqrt{4\pi}$ the two particles become free and for $\beta>\sqrt{4\pi}$ their interaction is repulsive. 
In both cases there is no longer any bound state between a soliton--antisoliton pair and the energy gap equals the sum of their masses, i.e., 
\begin{equation}
    \Delta_\mathrm{sG} = 2 M_s  \quad (\beta\geq\sqrt{4\pi}). 
\end{equation}

\subsection{Massive Schwinger model} 
\label{app:MS}

The construction of the bosonised version of the mS model follows Refs.~\cite{schwinger_gauge_1962, coleman_more_1976}. 
The construction of the bosonised massless Schwinger model in finite space was explained in detail in Ref.~\cite{iso_hamiltonian_1990}. 
Similarly, the construction of the massive model is discussed in Ref.~\cite{Stuart2014}. 

The Hamiltonian of the bosonised mS model can be constructed similarly to the sG Hamiltonian with the difference, that the free Hamiltonian corresponds to a massive instead of massless bosonic field and the frequency parameter is set to $\beta = \sqrt{4\pi}$. 
More specifically, we have
\begin{equation}
    H_{\mathrm{mS}} = H_{0} + \lambda V
\end{equation}
with $H_{0}$ as the free massive boson Hamiltonian and $V$ the cosine interaction 
\begin{align}
    H_{0} & =\frac{1}{2}\int_{0}^{L}\mathrm{d}x\,\left(\Pi^{2}+(\partial_{x}\Phi)^{2}+M^{2}\Phi^{2}\right) ,\\
    V & = \int_{0}^{L} \mathrm{d}x \, \threeOrd{\cos(\sqrt{4\pi}\Phi - \theta)}.
\end{align}
The normal ordering of the interaction operator, denoted by $\threeOrd{\cdot}$, refers now to the massive momentum modes of $H_{0}$. 
The boson mass parameter is
\begin{equation}
    M=\frac{e}{\sqrt{\pi}},
\end{equation}
where $e$ is the fermion electric charge. 
The coupling parameter is 
\begin{equation}
    \lambda=-\frac{mM}{4\pi}\mathrm{e}^{\gamma}
\end{equation}
with the Euler constant $\gamma \approx 0.57721$. 
Here we have used the expression for the thermodynamic limit, i.e., for $L \rightarrow \infty$. 
As before, $L$ is the system size, and we choose units $\hbar=1, c=1$. 
The mode expansion for the massive Schwinger model is 
\begin{equation}
    \Phi(x)=\sum_{k \in \mathbb{Z}}f_{k}(x)\tilde{\Phi}_{k}
    \label{eq:massive_mode_expansion_phi}
\end{equation}
with 
\begin{equation}
    f_{k}(x)=\frac{1}{\sqrt{L}}\mathrm{e}^{\mathrm{i}p_{k}x}
\end{equation}
and momenta 
\begin{equation}
    p_{k}=2\pi k/L,\;k \in \mathbb{Z}.
\end{equation}
Note that there is no winding mode $M$ and the zero momentum mode is not singled out (and rescaled) as in the sine-Gordon model. 
Written in momentum space, the free Hamiltonian becomes
\begin{equation}
    H_{0}=\sum_{k}\omega_{k}A_{k}^{\dagger}A_{k} + \text{const.},
\end{equation}
where the harmonic mode frequencies $\omega_{k}$ are now given by the relativistic dispersion relation 
\begin{align}
    \omega_{k} &= E_{k} = \sqrt{p_{k}^{2}+M^{2}}.
    \label{eq:massive_disp_rel}
\end{align}
The basis vector are characterised by the sequence $\{n_{k}\}$ of occupation numbers of the modes 
\begin{equation}
    |\Psi\rangle = |\{n_{k}\}\rangle=\prod_{k \in \mathbb Z} \frac{A_{k}^{\dagger n_{k}}}{\sqrt{n_{k}!}}|0\rangle,\qquad k\in\mathbb{Z},n_{k}\in\mathbb{N}.
    \label{eq:massive_basis_states}
\end{equation}
Exponential operators of the $\Phi$ field are defined in the previous form according to
\begin{equation}
    V_{\alpha}(x)=\mathrm{e}^{\mathrm{i}\alpha\Phi(x)},
\end{equation}
but given the different mode expansion, the matrix elements of normal ordered vertex operators between arbitrary pairs of basis states now become
\begin{align*}
    \begin{split}
    \langle\Psi'| \, \threeOrd{V_{\alpha}(x)} |\Psi\rangle & = \langle\{n'_{k}\}| \threeOrd{\mathrm{e}^{\mathrm{i}\alpha\Phi(x)}} |\{n_{k}\}\rangle \\
    & = \prod_{k \in \mathbb Z} F(n'_{k},n_{k};z,-z^{*})
    \end{split}
\end{align*}
with the product including the zero momentum mode. 
The function $F$ is given by the previous expressions in Eq.~\eqref{eq:F} and Eq.~\eqref{eq:F2} with $z = \mathrm{i}\alpha f_{k}(x)/\sqrt{2\omega_{k}}$, but $\omega_{k}$ is now given by Eq.~\eqref{eq:massive_disp_rel}. 
Therefore, we find
\begin{align*}
    \langle\Psi'| \threeOrd{V_{\alpha}(x)} |\Psi\rangle & = \mathrm{e}^{\mathrm{i}\left[\sum_{k}(n'_{k}-n_{k})p_{k}\right]x} \, \prod_{k \in \mathbb Z} G_{n'_{k},n_{k}} \left( \frac{\alpha^{2}}{2\omega_{k}L} \right).
\end{align*}
The interaction Hamiltonian is constructed as a sum of spatially integrated vertex operators including the additional phase factors $\mathrm{e}^{\pm\mathrm{i}\theta}$, i.e.,
\begin{align*}
    V & = \int_{0}^{L} \mathrm{d}x \, \threeOrd{\cos(\sqrt{4\pi}\Phi - \theta)} \\
    & = \frac{1}{2}\sum_{\sigma=\pm}\mathrm{e}^{-\mathrm{i}\sigma\theta} \int_{0}^{L} \mathrm{d}x \, \threeOrd{V_{\sigma\sqrt{4\pi}}(x)}
\end{align*}
with matrix elements 
\begin{align*}
    \int_{0}^{L} \mathrm{d}x \; \langle\Psi'| \threeOrd{V_{\pm\sqrt{4\pi}}(x)} |\Psi\rangle & = L \delta_{\sum_{k}(n'_{k}-n_{k})k,0} \\
    & \times \prod_{k \in \mathbb Z} G_{n'_{k},n_{k}}\left(\frac{2\pi}{\omega_{k}L}\right).
\end{align*}
Using Eq.~\eqref{eq:massive_disp_rel}, the argument of the functions $G$ is 
\begin{equation}
    \frac{2\pi}{\omega_{k}L}=\frac{1}{\sqrt{k^{2}+\left(\frac{ML}{2\pi}\right)^{2}}}.
\end{equation}

\section{Tensor network details}

\subsection{Introduction}

Tensor networks are representations of quantum many-body systems originally developed in the context of condensed matter physics. 
They consist of networks of interconnected tensors, whose structure typically resembles the underlying physical lattice and moreover the entanglement structure of the system. 
Given the decomposition in terms of a contracted network of smaller tensors, TNs are highly efficient representations of the wave function. 
While the Hilbert space is exponentially large in the number of constituents, tensor networks overcome this limitation with only a polynomial scaling by targeting the low-entangled states. 
Their expressive power is mostly controlled and limited by the amount of entanglement the TN can capture. 
In one spatial dimension, the most prominent family of TNs is known as the \emph{matrix product state} (MPS). 
It is the TN structure we will employ for the study of $(1+1)d$ quantum field theories.

In our method, we employ two principal algorithms -- the \emph{density matrix renormalisation group} method (DMRG)~\cite{White1992,White1993} to obtain ground states and low-energy excited states, and the \emph{time-dependent variational principle} (TDVP)~\cite{Haegeman2011,Haegeman2014,PhysRevLett.100.130501} to compute the time evolution under a quench. 
In this appendix, we will present a detailed explanation of the full TN construction  including the simulation algorithms used.


\subsection{General setup}
\label{sub:generalSetup}

\subsubsection{Matrix product states}
\label{ssub:matrixProductStates}

Matrix product states represent the state vector of the quantum system in the full many-body Hilbert space. 
The probability amplitudes of the wave function are encoded in and recovered by a multiplication of the individual MPS tensors
\begin{align}
    \mathcal S_{n_1, n_2, \hdots, n_N} = \sum_{\lbrace \alpha \rbrace} \mathcal M_{\alpha_L, n_1}^{\alpha_1} \mathcal M_{\alpha_1, n_2}^{\alpha_2} \hdots \mathcal M_{\alpha_{N-1}, n_N}^{\alpha_R}.
    \label{eq:waveFunctionDecomposition_1}
\end{align}
The variables $n_i$ correspond to the physical degrees of freedom of the system. 
In order to simulate the quantum field theories in the paper, we use a momentum-space TN. 
Unlike the typical formulation, in which each MPS tensor represents one lattice site (or a small collection of multiple lattice sites), our MPS consists of tensors for each individual momentum modes. 
Therefore, the first control parameter in the system is the maximal momentum, labelled by the maximal wave number $k_\mathrm{max}$. 
It determines the number of momentum modes, and hence the length of the MPS. 
Since we employ periodic boundary conditions in real space, the possible momentum modes are given by
\begin{align}
    k \, \in \, \lbrack -k_\mathrm{max}, \hdots, +k_\mathrm{max} \rbrack.
\end{align}
Each site of the MPS corresponds to one momentum mode, labelled with wave number $k$. For some models we have to distinguish between the zero-mode with $k = 0$ and all other modes with $k \neq 0$. 
All non-zero modes are harmonic modes of a quantum harmonic oscillator, as described in Sec.~\ref{subsub:modeExpansionFockBasis}. 
In our truncated Fock basis, they are assigned a local bosonic Hilbert space $\mathcal H_k$ with physical dimension $d_k = \operatorname{dim}(\mathcal H_k)$, which can be occupied with a maximal number of $n(k)$ bosons. 
In this Fock basis each physical Hilbert space is labelled by the occupation numbers, i.e., $\mathcal H_k = \lbrace \ket{0}, \ket{1}, \hdots, \ket{n(k)} \rbrace$, and therefore has a dimension of $d_k = n(k) + 1$.
For the sG model the $k = 0$ mode is described by a different algebra as explained in Sec.~\ref{sec:SG_H0}. 
It is however bounded by an additional control parameter $n_\mathrm{ZM}$. The truncation of the formally infinite bosonic Hilbert spaces $\mathcal H_k$ is a second control parameter in the simulations. 
Imitating the energy-based truncation in the regular TCSA procedure, we implemented a non-uniform restriction of the maximal occupation numbers as visualised in Fig.~\ref{fig:truncMethods_1}.
\begin{figure}[ht]
    \centering
    \includegraphics[scale = 1.0]{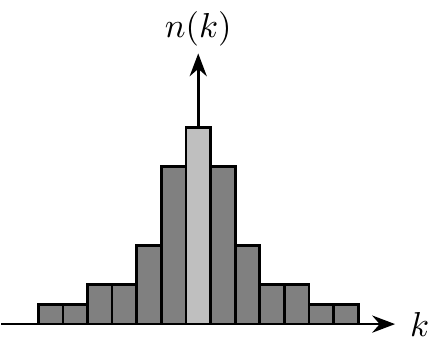}
    \caption{Truncation of the formally infinite bosonic Hilbert spaces for the harmonic modes. The non-uniform slope imitates the energy-based truncation in the regular TCSA method. Occupation of the zero-momentum mode is controlled by an additional parameter $n_\mathrm{ZM}$.}
    \label{fig:truncMethods_1}
\end{figure}
While the two truncations are not directly comparable, our procedure generally includes a much higher number of individual basis states. 
In order to generate the occupation number profile, we introduce an additional parameter $n_\mathrm{max}$, that limits the maximal occupation of non-zero modes in the system. 
A uniform distribution will have exactly $n_\mathrm{max}$ bosons per mode, however, in our generally used non-uniform profile the actual momentum-mode occupation numbers are then given by
\begin{align}
    n(k) = \left\lfloor \frac{n_\mathrm{max}}{\vert k \vert} \right\rfloor
    \label{n(k)}
\end{align}
for $k \neq 0$. 
Choosing $k_\mathrm{max} = 6$ and $n_\mathrm{max} = 8$ hence results in a MPS setup with modes as specified in Tab.~\ref{tab:occupationNumberProfile}.
\begin{table}[ht]
    \centering
    \begin{tabular}{c | c c c c c c c c c c c c c}
         $k$    & $-6$ & $-5$ & $-4$ & $-3$ & $-2$ & $-1$ & 0 & $+1$ & $+2$ & $+3$ & $+4$ & $+5$ & $+6$ \\
         \hline
         $n(k)$ & 1 & 1 & 2 & 2 & 4 & 8 & $n_\mathrm{ZM}$ & 8 & 4 & 2 & 2 & 1 & 1
    \end{tabular}
    \caption{Symmetric, non-uniform occupation number profile for the truncation of the local Hilbert spaces.}
    \label{tab:occupationNumberProfile}
\end{table}
For the zero-momentum mode we introduce a separate truncation parameter $n_\mathrm{ZM}$, that will be discussed individually for the 
sG and mS model.

Finally, we can set up the momentum-space MPS with the parameters $k_\mathrm{max}$, $n_\mathrm{max}$ and $n_\mathrm{ZM}$. 
An example for a system with $k_\mathrm{max} = 3$ is given in Fig.~\ref{fig:symmetricMPS_1}.
\begin{figure}[ht]
    \centering
    \includegraphics[scale = 0.975]{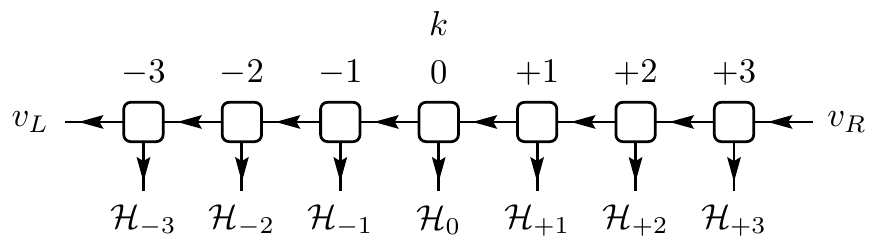}
    \caption{Momentum-space MPS. Each site corresponds to one momentum mode with wave number $k$ and physical Hilbert space $\mathcal H_k$. We consider a finite chain, so that the virtual bond indices $v_L$ and $v_R$ are one-dimensional. Arrows are only relevant for symmetric tensors, see Sec.~\ref{ssub:ConservationOfMomentum}.}
    \label{fig:symmetricMPS_1}
\end{figure}
Since vertical indices correspond to the physical degrees of freedom, they are called physical indices. 
Horizontal indices are called virtual indices, they are the ones which connect the individual tensors and mediate quantum correlations between them. 
The maximal dimension of all virtual indices is denoted as the \emph{bond dimension} $\chi_\text{max}$ of the MPS. 
For an MPS in the truncated Hilbert space defined by a fixed set of $k_\mathrm{max}$, $n_\mathrm{max}$ and $n_\mathrm{ZM}$, the bond dimension is the most relevant tuning parameter in the system. 
It controls the amount of entanglement the TN can capture and in turn the level of approximations to the true target state. 
So far the TN ansatz is non-symmetric, meaning that the individual three-index tensors are regular dense numerical arrays.
Arrows in Fig.~\ref{fig:symmetricMPS_1} have no meaning, until we impose quantum number conservation in the next section. 
Since we work with a finite MPS, the left and right virtual edge vector spaces are one-dimensional and can be set to the vacuum, i.e., $v_L = v_R' = \ket{0}$. 
In the next section we will generalise the concepts introduced above to operators, such that we can implement the Hamiltonians for the models of interest.


\subsubsection{Matrix product operators}
\label{ssub:matrixProductOperators}

Tensor networks are not only suited to represent a quantum state of a many-body system, they can also efficiently represent operators in a similar fashion. 
An operator $\hat O$ acting on the state vector $\ket{\psi}$ represented by an MPS can be implemented by a \emph{matrix product operator} (MPO). 
The transition amplitudes are encoded in and recovered by a multiplication of the
individual MPO tensors
\begin{align}
    \mathcal O_{n_1, n_2, \hdots n_N}^{n_1^\prime, n_2^\prime, \hdots, n_N^\prime} = \sum_{\lbrace \alpha \rbrace} \mathcal T_{\alpha_L,n_1}^{n_1^\prime,\alpha_1} \mathcal T_{\alpha_1,n_2}^{n_2^\prime,\alpha_2} \hdots \mathcal T_{\alpha_{N-1},n_N}^{n_N^\prime,\alpha_R}.
\end{align}
Each local tensor now has two physical indices alongside the two virtual ones, that connect them among each other. 
A visualisation in tensor network notation is shown in Fig.~\ref{fig:symmetricMPO_1}.
\begin{figure}[ht]
    \centering
    \includegraphics[scale = 0.975]{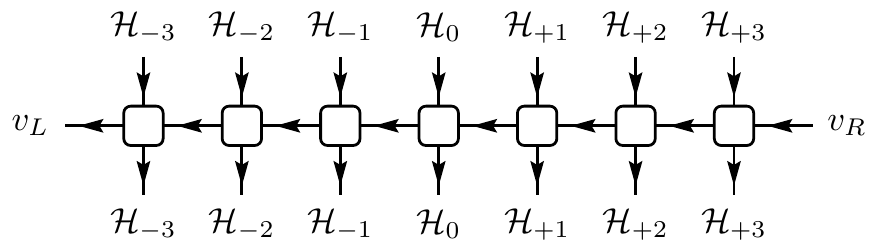}
    \caption{Momentum-space MPO. The momentum modes have the same $k$ labels as well as the same physical Hilbert spaces $\mathcal H_k$ as the MPS in Fig.~\ref{fig:symmetricMPS_1}. 
    For a finite system the virtual bond indices $v_L$ and $v_R$ are one-dimensional. 
    Arrows are only relevant for symmetric tensors, see Sec.~\ref{ssub:ConservationOfMomentum}.}
    \label{fig:symmetricMPO_1}
\end{figure}
The bond dimensions of the MPO, i.e., the dimensions of the virtual indices depend on the operator it encodes. 
While short-range Hamiltonians for quantum lattice systems can be implemented with a bond dimension $\chi_\mathrm{MPO}$ in the order of ten, our momentum-space Hamiltonian MPO will require larger spaces in the order of several hundreds.\\

The multiplication of an MPO with an MPS is one of the fundamental operations. 
The outcome is again an MPS with a larger bond dimension than the original one, since the individual virtual bond dimensions get multiplied. 
In order to keep the bond dimension finite and the simulations manageable, a \emph{truncation} of the final MPS is essential. 
We will comment on the most significant truncation of our method in Sec.~\ref{sub:densityMatrixRenormalizationGroup}.

\subsubsection{Computation of expectation values}

Finally, TNs and in particular MPS offer a very convenient way to compute expectation values. 
The network architecture is flexible with regard to single- and multi-site observables, as well as global operators in the form of MPOs. 
For instance, the energy expectation value $E = \braket{\psi \vert H \vert \psi}$ can be easily computed, as all parts can be conveniently expressed in TN form. 
It is visualised in Fig.~\ref{fig:expectationValues_MPS_1}, where the Hamiltonian MPO is placed between the state vector $\ket{\psi}$ and its conjugate. 
\begin{figure}[ht]
    \centering
    \includegraphics[scale = 0.975]{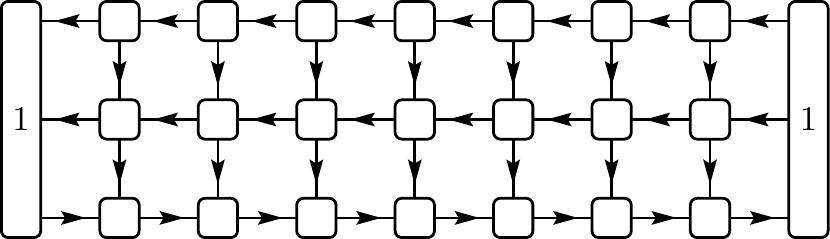}
    \caption{Calculation of MPO expectation values in the TN representation. Here we assume the MPS to be properly normalised to $\braket{\psi \vert \psi} = 1$, so that there is no denominator to the expectation value.}
    \label{fig:expectationValues_MPS_1}
\end{figure}
The outermost MPS and MPO indices are one-dimensional and the full network can be terminated by identity boundary tensors.


\subsubsection{Conservation of momentum}
\label{ssub:ConservationOfMomentum}

For both models discussed in the paper, the sine-Gordon model and the massive Schwinger model, translation invariance of the Hamiltonian implies conservation of the total momentum. 
This is implemented by a global symmetry constraint, which forces the MPS to be in a sector with fixed total momentum, given by
\begin{align}
    P_\mathrm{tot} :=  \frac{2\pi}{L} \sum_{k = -k_\mathrm{max}}^{+k_\mathrm{max}} k n_k = \text{const}.
\end{align}
Fortunately, tensor networks offer a convenient and very efficient way to implement physical symmetries (such as $\mathbb Z_n$, $U(1)$ or $SU(2)$) at the level of the tensors~\cite{STN_0,STN_1,Schmoll2018,Silvi2019}. 
Besides fixing the symmetry sector for the full many-body quantum state, the symmetry preservation leads to a sparse internal block structure of the tensors, which in turn results in a computational speed-up. 
Upon imposing symmetry constraints, all tensor indices are labelled by the irreducibly representations of the symmetry group according to
\begin{align}
    \mathbb V \cong \bigoplus_a \mathbb V^a_{d_a}.
    \label{eq:vectorSpaceDecomposition}
\end{align}
Here $a$ labels the irreducible representations with vector spaces $\mathbb V^a$, and $d_a$ denotes the degeneracy of each element in Eq.~\eqref{eq:vectorSpaceDecomposition}. 
The total vector space dimension is given by
\begin{align}
    \text{dim}(\mathbb V) = \sum_a d_a \cdot \text{dim}(\mathbb V^a).    
\end{align}
For the underlying $U(1)$ symmetry used in our momentum-conserving MPS, the irreducible representations have no intrinsic degeneracy, such that $\text{dim}(\mathbb V^a) = 1$ for all $a$. 
The direction of a tensor index, denoted by the arrows in the TN diagrams, determines how the symmetry group acts on the index~\cite{STN_0}. 
In practise it means, that an outgoing index can only be contracted with an incoming index, and vice versa. 
In order to construct the wave function in a fixed momentum sector, the full MPS is constructed out of momentum-conserving tensors. 
This can be achieved by a modified $U(1)$ symmetry. 
A regular $U(1)$ symmetry corresponds to particle number conservation, which is the underlying symmetry for the Fock space of the momentum modes. 
However, the occupation numbers have to be scaled by the corresponding momentum, so that global momentum conservation can be imposed by a weighted $U(1)$ symmetry. 
As an example we consider a local MPS tensor for momentum $k = -3$ and $n(-3) = 4$. The physical Hilbert space is then labelled by the basis states
\begin{align}
    \mathcal H_{-3} = \big\lbrace \ket{0_1}, \ket{-3_1}, \ket{-6_1}, \ket{-9_1}, \ket{-12_1} \big\rbrace,
\end{align}
where we have used $\ket{a_{d_a}}$ to denote quantum numbers and their degeneracy. 
This procedure is also applied to the virtual indices of the MPS tensors, such that they are labelled by momentum quantum numbers (QN), too. 
The possible QN configurations on the virtual indices can be constructed from the physical spaces and the outermost virtual spaces. 
For a total zero-momentum state we set $v_L = v_R = \ket{0_1}$, while a state in a sector with arbitrary momentum $P = 2\pi k/L$ can be generated by \mbox{$v_L = \ket{0_1}$} and \mbox{$v_R' = \ket{k_1}$}. 
In the next section we will briefly outline the density matrix renormalisation group procedure used to compute ground and low-energy excited states for one-dimensional TNs.


\subsection{Density matrix renormalisation group}
\label{sub:densityMatrixRenormalizationGroup}

At the heart of our method lies the DMRG algorithm, which is used to optimise the MPS coefficients and obtain ground states and low-energy excited states~\cite{Schollwoeck2011,Stoudenmire2012}. 
To this end we employ a regular two-site DMRG code, with additional modifications required due to obstacles coming from the conservation of momentum. 
A two-site version of DMRG is beneficial for symmetric simulations, as is provides better automatic handling of quantum numbers. 

Assuming a Hamiltonian given as an MPO, the idea of DMRG is to find a quantum state vector $\ket{\psi_0}$ in the form of an MPS, that minimises the energy
\begin{align}
    E_0 = \text{min}_{\ket{\psi}} \frac{\braket{\psi \vert H \vert \psi}}{\braket{\psi \vert \psi}}.
    \label{eq:optimalDMRG_1}
\end{align}
We search for the optimal state starting from a random initial MPS and progressively optimising the tensor coefficients to lower the state's energy expectation value. 
Throughout the full procedure, the MPS can be kept in a suitable canonical form and such that it is normalised with $\braket{\psi \vert \psi} = 1$. 
Importantly, the optimisation in Eq.~\eqref{eq:optimalDMRG_1} does not need to be performed globally. 
Instead, it is performed sequentially, fixing everything except two neighbouring tensors with wave numbers $(k,k+1)$ and optimising the parameters of these tensors at a time. 
This process is visualised in Fig.~\ref{fig:DMRG_1}.
\begin{figure}[ht]
    \centering
    \includegraphics[scale = 0.975]{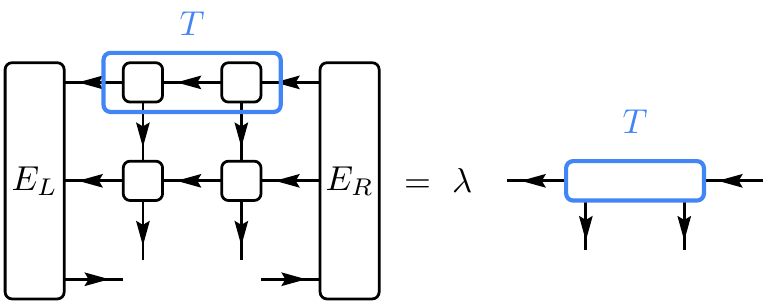}
    \caption{Two-site update in the DMRG algorithm. 
    The parts on the left and right of the two active sites are exactly represented by effective environment tensors $E_L$ and $E_R$. 
    The updated two-site tensor $T$ is obtained by solving this equation for the eigenstate with lowest eigenvalue $\lambda$.}
    \label{fig:DMRG_1}
\end{figure}
The finite MPS chain is terminated by identity boundary tensors. 
For each two-site update, the parts left of mode $k$ and right of mode $k+1$ can be exactly and efficiently represented by environment tensors $E_L$ and $E_R$ respectively. 
The updated two-site tensor $T$ is then given by the eigenstate of the equation in Fig.~\ref{fig:DMRG_1} with lowest eigenvalue. 
Since two sites are updated simultaneously, tensor $T$ needs to be decomposed back into separate MPS tensors. 
This is done with a \emph{singular value decomposition} (SVD)
\begin{align}
    \begin{split}
        \includegraphics[scale = 0.975]{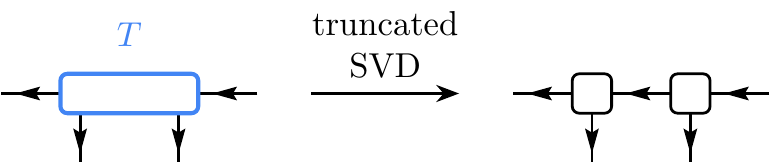}
    \end{split},
    \label{eq:twoSiteTruncation}
\end{align}
in which the singular values are truncated based on a maximal bond dimension $\chi_\text{max}$ or some truncation error criteria. 
Importantly, the truncation is done in each symmetry sector of the matrix, such that they can be automatically selected based on their singular weights. 
In each truncation, a truncation error can be defined as the norm of the discarded singular values with
\begin{align}
    \epsilon_D^2 = {\sum_{i > \chi} \lambda_i^2}.
\end{align}
Here, $\lambda_i$ are the singular values arranged in descending order and the subscript $\epsilon_D$ refers to the DMRG truncation parameter. 
In our simulations, we set an error-based truncation of $\epsilon_D = 10^{-6}$ together with a strict upper limit for the bond dimension.\\

The two-site update procedure is repeated by sweeping from left to right and back through the MPS, until the full state converges. 
In a left (right) sweep the right (left) environment tensors are updated by absorbing the column of tensors for mode $k+1$ ($k$) into the right (left). 
After convergence, we compute the energy variance according to
\begin{align}
    \begin{split}
        \ket{\phi} &= H\ket{\psi_0} - E\ket{\psi_0} ,\\
        \text{var}(E) &= \braket{\phi \vert \phi} = \braket{\psi_0 \vert H^2 \vert \psi_0} - E^2,
    \end{split}
\end{align}
with $E = \braket{\psi_0 \vert H \vert \psi_0}$. 
The variance determines how close the optimised MPS is to an eigenstate of the Hamiltonian. 
It can diagnose insufficient bond dimension and, more importantly for our method, if the MPS virtual indices contain the relevant quantum numbers to best represent the target state.


\subsubsection{Low-energy excited states}

Excited states can be of two kinds. 
The first kind is a regular eigenstate in the same momentum sector as the ground state. 
In this case the DMRG code for ground states can be extended to compute those low-lying excited states, too. 
To this end, the algorithm is set up to find the state with lowest energy, that is orthogonal to lower energy states~\cite{Stoudenmire2012}, such as the ground state. 
While the first excited state is found by orthogonalisation to the ground state, the second excited state needs to be orthogonal to both the ground and the first excited state, and so on.\\

The second kind of excited state is an eigenstate in a different momentum sector. 
In this case, we can use a regular ground state DMRG simulation, where the MPS boundary vector spaces $v_L$ and $v_R$ are set to target the corresponding total momentum sector.


\subsubsection{Global subspace expansion}

In regular DMRG simulations with symmetry-preserving tensors, a two-site update is advantageous due to the automatic selection of relevant symmetry sectors during the truncation in Eq.~\eqref{eq:twoSiteTruncation}. 
Moreover, it is a convenient way to increase the bond dimension on the virtual indices by introducing new, previously absent symmetry sectors. 
The algorithm is hence able to adapt the set of quantum numbers dynamically based on their weights in the singular value spectrum. 
Unfortunately, in our setting, the different wave vectors $k$ lead to non-uniform Hilbert spaces $\mathcal H_k$, which impede the automated handling of quantum numbers. 
While less relevant sectors can still be excluded in the truncation, a regular two-site DMRG implementation is no longer able to introduce all required quantum numbers automatically. 
A local subspace expansion~\cite{Hubig2015} is also not possible, since the determination of new sectors that need to be included expands beyond the two-site region to a global problem. 
Therefore, we directly implement a global subspace expansion based on the full Hamiltonian. 
Starting from a random state vector $\ket{\psi}$ with some initial set of quantum numbers on the virtual bonds, we multiply it with the Hamiltonian and use a moderate truncation to compress the resulting MPS. 
This introduces new QNs before the DMRG sweeps are performed. 
Once sufficiently many QNs are present, the two-site optimisation algorithm itself is able to introduce new ones, so that the global subspace typically has to be performed only once at the start of the procedure.


\subsection{Time-dependent variational principle}

Conveniently, the MPS framework can also be used to perform time evolution of the full many-body quantum state. 
For this task we have implemented a two-site TDVP algorithm~\cite{Haegeman2011,Haegeman2014}. The algorithm itself is similar to two-site DMRG presented above. 

However, instead of solving Eq.~\eqref{fig:DMRG_1} for the eigenstate with lowest energy, a computation of the evolution operator by exponentiation is used to solve the time-dependent Schrödinger equation locally. 
Moreover, an additional backwards evolution step on the second site (the left (right) site in a left (right) sweep) is necessary to restore it to original time, before proceeding with the sweep~\cite{Haegeman2014}. 
Similarly to DMRG, a two-site approach is used to choose the relevant symmetry sectors dynamically. 
However, the problem of introducing new QNs in the first place persists also in the TDVP code. 
For this reason we employ a global Krylov subspace expansion technique~\cite{Yang2020} prior to the actual TDVP sweeping.

Here we will briefly outline the main steps of the time evolution algorithm in order to introduce the relevant truncation parameters. 
Evolving a state vector $\ket{\psi(t)}$ to $\ket{\psi(t + \Delta t)}$ is performed by the following steps~\cite{Yang2020}.
\begin{itemize}
    \item Construct $k$ different global Krylov vectors by successively applying the Hamiltonian MPO to $\ket{\psi(t)}$, controlled by a truncation error $\epsilon_K$. 
    \item Expand the basis of $\ket{\psi(t)}$ by a portion of the $k$ global Krylov vectors. This step is controlled by a truncation error $\epsilon_M$. 
    \item Perform a regular two-site TDVP sweep from left to right and back. The truncation error in the SVD is controlled by $\epsilon_T$.
\end{itemize}
Overall, there are five different refinement parameters in the TDVP algorithm. 
The smallest time step $\Delta t$, the number of Krylov vectors $k$ and the three truncation error parameters. 
For the two-site TDVP update we choose a time step \mbox{$\Delta t = 2\cdot 10^{-2}$} and an error-based truncation \mbox{$\epsilon_T = 10^{-4}$} together with a hard bond dimension limit $\chi = 2500$. 
In our setting, the parameters $k = 2$, \mbox{$\epsilon_K = 10^{-8}$} and \mbox{$\epsilon_M = 10^{-10}$} turned out to be reasonable to achieve sufficient accuracy, in turn leading to a high computational cost. 
To keep the simulations manageable, a upper bond dimension limit of $\chi = 3000$ was required for the construction of the Krylov vectors.


\subsection{Implementation of the Hamiltonian}

The Hamiltonian is constructed as two independent parts, a non-interacting and an interacting part. 
The non-interacting part can be conveniently constructed using the regular formalism for local Hamiltonians based on finite state machines~\cite{finiteStateMachineMPO1,finiteStateMachineMPO2}. 
For the interacting part however, we will introduce a supporting TN structure called the Kronecker-Delta MPS. 
As explained in the main text and the method section, it will be an important tool to ensure a fully momentum-preserving construction of the MPO.


\subsubsection{Non-interacting Hamiltonian}

The non-interacting part of the sine-Gordon and the massive Schwinger model can be constructed on the same footing. 
We start with the general Hamiltonian in Eq.~\eqref{eq:normalOrderedFreeHamiltonianSG} consisting of harmonic modes for non-zero momentum modes and a separate zero-mode contribution, i.e.,
\begin{align}
    H_0 = \sum_{k \in \mathbb Z^*} E_k A_k^\dagger A_k + \mathcal O
    \label{eq:KleinGordonHamiltonian}
\end{align}
with the generally massive energy-momentum relations
\begin{align}
    p_k = \frac{2\pi}{L} k, \hspace{0.5cm} E_k = \sqrt{p_k^2 + M^2}.
\end{align}
Here $M = e/\sqrt{\pi}$ is the effective boson mass and $e$ the electric charge. 
The zero-mode operator for the sG model is given by
\begin{align*}
    \mathcal O_\mathrm{sG} = \frac{1}{2} L \tilde\Pi_0^2,
\end{align*}
that can be represented in the $(2n_\mathrm{ZM} + 1)$-dimensional truncated eigenspace of the $\tilde\Pi_0$ operator according to Eq.~\eqref{eq:zero-mode_eigstates_1} as 
\begin{align*}
    \tilde\Pi_0 = \frac{1}{R} \begin{pmatrix} - n_\mathrm{ZM} & 0 & \hdots & 0 & 0 \\
    0 & - n_\mathrm{ZM} + 1 & 0 & \ddots & 0 \\
    \vdots & 0 & \ddots & 0 & \vdots \\
    0 & \ddots & 0 & +n_\mathrm{ZM} - 1 & 0 \\
    0 & 0 & \hdots & 0 & +n_\mathrm{ZM}
    \end{pmatrix}.
\end{align*}
In the mS model the zero-mode is an additional harmonic mode described by
\begin{align*}
    \mathcal O_\mathrm{mS} = M A_0^\dagger A_0,
\end{align*}
that acts on a $(n_\mathrm{ZM} + 1)$-dimensional Fock space. 
In the momentum-space TN representation, Eq.~\eqref{eq:KleinGordonHamiltonian} is just a sum of local terms, that can be implemented by an MPO with bond dimension two. 
The MPO tensors for the zero and non-zero modes are shown in Fig.~\ref{fig:mpo_H0_1}.
\begin{figure}[ht]
    \centering
    \includegraphics[scale = 1.0]{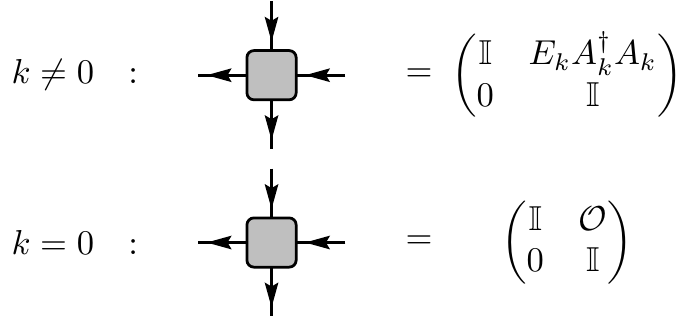}
    \caption{MPO form of the non-interacting Hamiltonian $H_0$. The virtual bond dimension is $\chi_\mathrm{MPO} = 2$, the physical bond dimensions depend on the maximal occupation $n(k)$ of each mode.}
    \label{fig:mpo_H0_1}
\end{figure}
The left and right MPO tensors are terminated by edge vectors
\begin{align*}
    e_L = \begin{pmatrix} 1 & 0 \end{pmatrix} \hspace{0.5cm} \text{and} \hspace{0.5cm} e_R = \begin{pmatrix} 0 \\ 1 \end{pmatrix}. 
\end{align*}
The ground state of $H_0$ is given by the vacuum state with no bosonic occupations and the zero-eigenstate for the $\mathcal O$ operator.


\subsubsection{Kronecker-Delta MPS}

Before continuing with the interacting part of the Hamiltonian, we introduce a supporting TN structure called the Kronecker-Delta MPS. 
It is one of the most important pieces in the TN construction, as it automatically generates \emph{all} possible interaction terms of the vertex operator while ensuring momentum conservation in the construction. 
The origin and necessity of this global symmetry constraint has been explained in Sec.~\ref{app:modelDefinition}, here we will present its actual implementation. 
The Kronecker-Delta MPS ($\delta$-MPS) has a similar form as the general momentum-space MPS ansatz in Fig.~\ref{fig:symmetricMPS_1}. 
However, instead of representing a wave function, it is constructed analytically to generate an \emph{equal-weight-one} superposition for all the possible QN combinations of its physical Hilbert spaces, that satisfy a global symmetry constraint. 
This means that the probability amplitudes for each possible QN combination in the $\delta$-MPS is exactly unity, which is going to weight the different interacting Hamiltonian terms evenly. 
The concrete physical Hilbert spaces will depend on the operator we wish to implement, for now we take them as general vector spaces with \emph{non-degenerate} QNs denoted by $\Delta_k$. 
The zero-mode hence only carries a single quantum number $\Delta_0 = \ket{0_1}$, because no momentum can be transferred away from or to this site. 
The virtual indices of the $\delta$-MPS are determined successively by \emph{all} possible combinations of quantum numbers, starting at the edges of the $\delta$-MPS and moving towards the center. 
Explicitly, they are constructed as $v_k^R = v_k^L \otimes \Delta_k$ for $k < 0$ and $v_k^L = \Delta_k^\prime \otimes v_k^R$ for $k > 0$, where $v_k^L$ and $v_k^R$ are the left and right vector spaces of the individual MPS tensors. 
Naturally we have $v_k^R = v_{k+1}^L$ and $v_k^L = v_{k-1}^R$. 
Importantly, degeneracies that arise in the fusion of the (modified) $U(1)$ quantum numbers have to be omitted, since they would modify the probability amplitudes in the $\delta$-MPS, therefore skewing the MPO construction. 
Finally at the central site, the left virtual, physical and right virtual vector spaces are combined into a total fusion space $\mathbb V_F$. Assuming $v_L = v_R' = \ket{0_1}$, it is given by
\begin{align}
    \begin{split}
        \mathbb V_F &= \Bigg\lbrack \underbrace{\left( \bigotimes_{k<0} \Delta_k \right)}_{v_0^L} \otimes \, \Delta_0 \otimes \underbrace{\left( \bigotimes_{k>0} \Delta_k \right)}_{\left( v_0^R \right)^\prime} \Bigg\rbrack_1.
    \end{split}
\end{align}
Here the notation $\lbrack \hdots \rbrack_1$ denotes, that all degeneracies are removed such that the QNs in $\mathbb V_F$ all have $d_a = 1$. 
The fusion space $\mathbb V_F$ is shown as the top (blue) index of the $\delta$-MPS in Fig.~\ref{fig:kronDeltaMPS_2}.
\begin{figure}[ht]
    \centering
    \includegraphics[scale = 0.975]{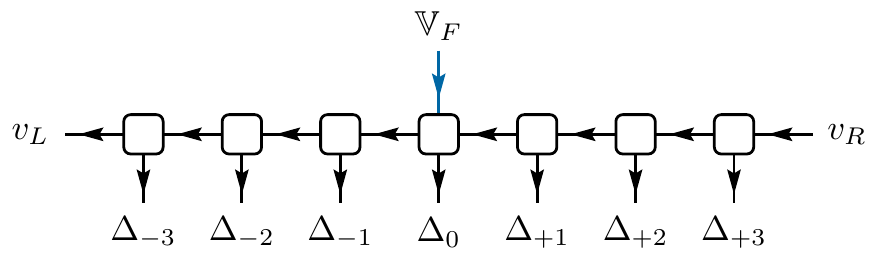}
    \caption{Kronecker-Delta MPS. The top blue index carries the (non-degenerate) product of all QNs on the physical vector spaces $\Delta_k$. Boundary vector spaces are typically set to be the vacuum.}
    \label{fig:kronDeltaMPS_2}
\end{figure}
With all the virtual and the fusion vector spaces constructed, the $\delta$-MPS carries \emph{only} non-degenerate quantum numbers on all the links. 
The individual symmetric blocks of the $\delta$-MPS are then set to one, such that it generates the desired equal-weight superposition of all the possible QN combinations.

For the construction of MPOs for operators that preserve the total momentum, it is required to restrict the total fusion space to $\mathbb V_F = \ket{0_1}$, which ultimately implements the global symmetry constraint. 
Since in this case $\mathbb V_F$ is one-dimensional with a trivial vector space, it can be omitted altogether.


\subsubsection{Interacting Hamiltonian}

The explicit forms of the interacting part of the sG and mS Hamiltonians have been presented in Sec.~\ref{app:SG} and Sec.~\ref{app:MS}. 
We will now construct the full interacting MPO for those models, by combining the local effect of the vertex operators with the global $\delta$-MPS. 
In general, the action of the vertex operator on each momentum mode can be implemented by a three-index tensor. 
The precise form of these tensors however might differ for zero and non-zero modes, similarly to the non-interacting Hamiltonian.\\

In the sG model, the zero-mode has to be treated differently from the rest. 
Its algebra is determined by Eq.~\eqref{eq:vertexOperatorZeroMode_sG}, that can be implemented by an off-diagonal matrix coupling eigenstates that differ by one (see Sec.~\ref{sub:spaceIntegratedInteractionOperator}). 
For the non-zero modes in the sG model and all modes in the mS model, the local three-index tensors are filled with matrix elements as defined in Eq.~\eqref{eq:G}. 
\begin{figure}[ht]
    \centering
    \includegraphics[scale = 1.0]{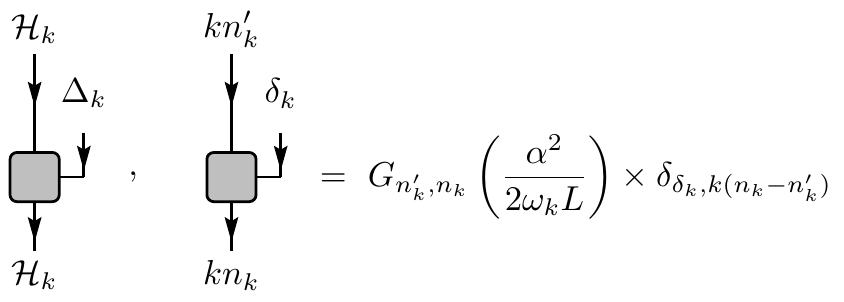}
    \caption{General form of the local three-index tensors as well as the individual tensor elements to encode the action of the vertex operator. The additional index $\delta_k$ takes all values of possible momentum transfer between the bra and ket index. Tensor elements are given by analytic functions, as defined in Eq.~\eqref{eq:G}.}
    \label{fig:mpo_H1_3}
\end{figure}
Here, $n_k^\prime$ and $n_k$ denote individual occupation numbers for the bra and ket indices. 
The third, auxiliary index is required in order to allow coupling of different occupations $n_k^\prime$ and $n_k$ (a regular $U(1)$-symmetric two-index local tensor is forced to be diagonal by symmetry, thus only allowing terms for which $n_k^\prime = n_k$).

Therefore, the third index determines the local momentum transfer between bra and ket indices, given by \mbox{$\delta_k = k(n_k - n_k^\prime)$}. 
The general form of the local three-index tensors as well as the individual tensor elements are shown in Fig.~\ref{fig:mpo_H1_3}. 
The momentum transfer index can take all possible values, i.e.,
\begin{align}
    \Delta_k &= \left\lbrack \mathcal H_k \otimes \mathcal H_k^\prime \right\rbrack_1 \label{eq:def_delta_k} \\
    &= \lbrace \ket{(-k n(k))_1}, \ket{(-k (n(k) - 1))_1}, \hdots, \ket{(+k n(k))_1}\rbrace. \nonumber
\end{align}    
A specific tensor element for fixed $n_k$, $n_k^\prime$ has then to be put into the corresponding position for the momentum difference. 

Collecting all tensors for the vertex operator, the first stage of the construction of the interacting MPO is shown in Fig.~\ref{fig:interactingMPO_2}.
\begin{figure}[ht]
    \centering
    \includegraphics[scale = 1.0]{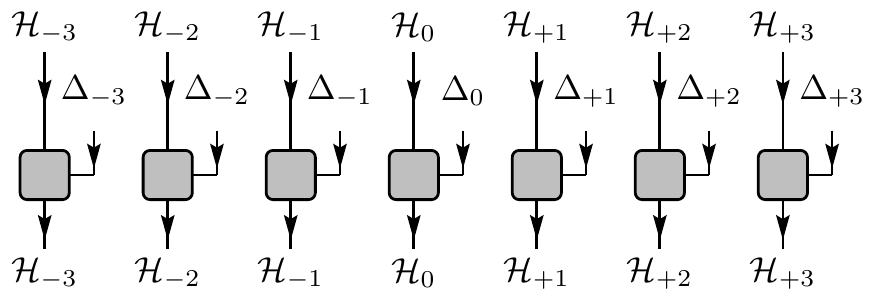}
    \caption{Product of local vertex operators as part of the construction of the interacting MPO.}
    \label{fig:interactingMPO_2}
\end{figure}
In the final step, the local momentum transfers, encoded into the auxiliary indices $\delta_k$ need to be coupled to ensure global momentum conservation (see Sec.~\ref{sub:spaceIntegratedInteractionOperator}). 
To this end we construct a $\delta$-MPS with physical indices given by $\Delta_k$. 
Ultimately, the full space-integrated vertex operator can be expressed as an MPO in Fig.~\ref{fig:interactingMPO_4}.
\begin{figure}[ht]
    \centering
    \includegraphics[width = \columnwidth]{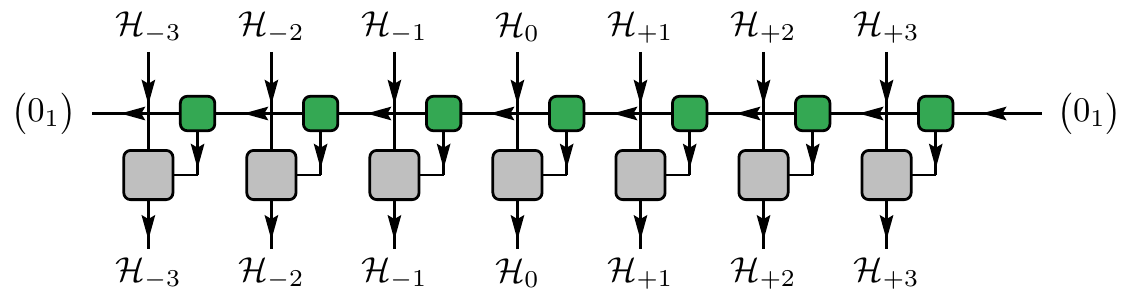}
    \caption{Full construction of the vertex operator in MPO form. By combining local interaction tensors with momentum carryover and a global projection operator in the form of an equal-weight superposition $\delta$-MPS, four index MPO tensors can be formed locally.}
    \label{fig:interactingMPO_4}
\end{figure}
The MPO bond dimension $\chi_\mathrm{MPO}$ of the vertex operator is therefore solely determined by the bond dimension of the Kronecker-Delta MPS. 
The cosine self-interaction can finally be constructed as the sum of two vertex operator MPOs with opposite frequencies, see Eq.~\eqref{eq:normalOrderedCosineInteraction}.


\subsection{Complexity analysis}

Let us now analyze the complexity of the tensor network algorithms. 
The algorithmic bottle neck for the DMRG as well as the TDVP routine is the evaluation of the local 2-site patch, as visualised in Fig.~\ref{fig:DMRG_1}. 
Denoting by $n(k)+1$ the physical dimension of the local mode $k$ and using $n=\max_k(n(k)+1)$, denoting by $\chi$ the MPS bond dimension and by $D$ the MPO bond dimension, then the left hand side of Fig.~\ref{fig:DMRG_1} can be computed with $O(\chi^3Dn+\chi^2D^2n^2)$ arithmetic operations. 
This translates to $O(k_{\mathrm{max}}(\chi^3Dn+\chi^2D^2n^2))$ arithmetic operations per DMRG, or similarly TDVP sweep. 
The parameters $k_{\mathrm{max}}$, $n_{\mathrm{max}}$ and $\chi$ are inputs to the algorithm, where $n(k) = \left\lfloor n_\mathrm{max}\vert k \vert^{-1} \right\rfloor\leq n_{\mathrm{max}}$ for $k\neq0$ and $n(0)=n_{\mathrm{ZM}}$. 
In contrast, $D$ is defined by the bond dimension of the $\delta$-MPS. 
To see this, note that $H=H_0+\lambda V$ where $H_0$ can be written as an MPO with bond dimension two and where $\lambda V$ can be written as an MPO with twice the bond dimension given by the $\delta$-MPS bond dimension $D_\delta$, such that $D=2D_\delta+2$. 
To obtain the desired complexity upper bound for both of the algorithms, it remains to bound $D_\delta$. 

To do so, we switch from the algebraic construction of the $\delta$-MPS to a combinatorial derivation. 
The algebraic approach has the advantage of being valid for any symmetry group. 
The combinatorial approach, in contrast, only works for groups isomorphic to $\mathbb{Z}_n$ (corresponding here to the truncated mode spaces) and products thereof. 
However, the combinatorial approach is more useful for the present analysis. 
More precisely, in what follows we are going to exploit a linear representation of addition to construct the $\delta$-MPS, which works for any ordering of modes and by construction -- because we represent addition on $\mathbb{Z}_K$ with $K$ corresponding to the relevant information -- has the optimal bond dimension. 

\begin{lemma}
The $\delta$-MPS can be constructed by an MPS with bond dimension $D_\delta = K\leq 4 k_{\mathrm{max}} n_{\mathrm{max}}+1$.
\end{lemma}

\begin{proof}
The task of the $\delta$-MPS is to constrain the vertex operator to those `momentum transfer configurations', which are compatible with global momentum conservation. 
The momentum transfer of each mode $k$ is encoded by the local Hilbert space $\Delta_k$ as defined in Eq.~\eqref{eq:def_delta_k}. 
We denote by $(\hat k_1,\dots,\hat k_l)$ with $l=2k_{\mathrm{max}}+1$ some ordering of the modes $(-k_{\mathrm{max}}, \dots, k_{\mathrm{max}})$. 
In particular, the hatted variables $\hat k_j$ do not need to be ordered monotonically in $j$. 
Then, we can explicitly construct the $\delta$-MPS with respect to that ordering as follows.
Let $\mathcal K$ be a Hilbert space of dimension $K$ with 
\begin{align}
    \begin{split}
        K-1 & =2\sum_{j=1}^l |\hat k_j n(\hat k_j)| = 4\sum_{{j'}=1}^{k_{\mathrm{max}}} j'\left\lfloor\frac{n_{\mathrm{max}}}{j'}\right\rfloor \\
        &\leq4\sum_{j'=1}^{k_{\mathrm{max}}} n_{\mathrm{max}} = 4k_{\mathrm{max}}n_{\mathrm{max}},
    \end{split}
\end{align}
which we will use for book-keeping the momentum transfer present in the current configuration.
To add some intuition in $K$ note the following. 
The $j$-th mode with momentum $\hat k_j$ participates in the momentum transfer by some 
\begin{align*}
    \begin{gathered}
        \delta_j \in \mathcal B_{j}, \\
        \mathcal B_j = \{-\hat k_j n(\hat k_j),-\hat k_j(n(\hat k_j)-1),\dots, +\hat k_jn(\hat k_j)\}.
    \end{gathered}
\end{align*}
The degree of freedom accounting for the global momentum transfer needs to encode all possible global momentum transfer values $\delta_g$, where the global momentum transfer for a specific configuration of local transfers is given by
\begin{align}
    \delta_g(\delta_1,\dots,\delta_l)=\sum_{j=1}^l \delta_j \,, \quad \delta_j \in \mathcal B_{j}.
    \label{eq:delta_g_values}
\end{align}
Note that for the extremal global momentum transfers it holds $\delta_g^{\mathrm{max}} = \sum_j |\hat k_j n(\hat k_j)|$ and similarly $\delta_g^{\mathrm{min}} = -\sum_j|\hat k_j n(\hat k_j)|$. 
Thus, $\delta_g$ can attain at most $K = \delta_g^{\mathrm{max}} - \delta_g^{\mathrm{min}} + 1$ different values and hence can be encoded by the Hilbert space $\mathcal K$. 
We label the corresponding canonical basis vectors by the set $[K]=\{1,\dots,K\}$. 

By construction, we associate the first canonical basis vector in $\mathcal K$ with the minimal global momentum transfer 
\begin{equation}
    \delta_g^{(1)} = \delta_g^{\mathrm{min}}=-\frac{K-1}{2}
\end{equation}
and similarly the $K$-th basis vector with the maximal global momentum transfer $\delta_g^{(K)}=\delta_g^{\mathrm{max}} = \frac{K-1}{2}$. 
Consequently, the $\frac{K+1}{2}$-th  canonical basis vector labels the global momentum conserving sector corresponding to $\delta_g^{\left(\!\frac{K+1}{2}\!\right)}=0$. 

Assume on site $j$ the accumulated momentum transfer of the modes $(1,\dots, j-1)$ is $\delta_{1,\dots,j-1}$ and the transfer on the $j$-th mode is $\delta_j$, then the accumulated momentum for the modes $(1,\dots, j)$ is equal to $\delta_{1,\dots,j} = \delta_{1,\dots,j-1} + \delta_j$. 
This motivates us to define the tensors 
\begin{align}
    \left[+^{(j)}\right] \in \{0,1\}^{[K]\times\mathcal B_j\times [K]} \;,\quad j=1,\dots, l
\end{align}
in terms of 
\begin{align}
    [+]\in\{0,1\}^{[K]\times[K]\times[K]}
\end{align}
as 
\begin{align}
    &{\left[+^{(j)}\right]}_{i,\alpha}^\beta=[+]_{
    \frac{K+1}{2}+i,\alpha}^\beta \;,  &&\alpha,\beta\in[K]\,,\; i\in\mathcal B_{j}\label{eq:def_plus1},\\[8pt]
    &[+]_{i,\alpha}^\beta = \delta_{\alpha+i,\beta}\;, &&\alpha,\beta, i\in[K]\,.
\end{align}
Here, we have  used the set theoretic notation $A^B=\{f:B\rightarrow A\}$ and identify each tensor in $\{0,1\}^{[K]\times\mathcal B_j\times[K]}$ as a map from the set of index labels $[K]\times\mathcal B_j\times[K]$ to its entries $\{0,1\}$. 

We can think of the $[+]$ tensor as a representation of $\mathbb{Z}_K$ on $\mathcal K$, such that $[+]_i : \ket{\alpha} \mapsto \ket{\alpha+i}$. 
Note that the additional $({K+1})/{2}$ in Eq.~\eqref{eq:def_plus1} is required in order to embed all $\mathcal B_j$'s into $[K]$ in a way consistent with the interpretation laid out in the paragraph after Eq.~\eqref{eq:delta_g_values}.

We are now ready to define the $\delta$-MPS with respect to the ordering $(\hat k_1,\dots,\hat k_l)$ and with respect to the momentum transfer configuration $(\delta_{1},\dots,\delta_{l})$, $\delta_j\in\mathcal{B}_{j}$, as 
\begin{align}
    \delta(\delta_{1},\dots,\delta_{l})=v^\dagger\cdot \left[+^{(1)}\right]_{\delta_1}\cdot\left[+^{(2)}\right]_{\delta_2}\cdots \left[+^{(l)}\right]_{\delta_l} \cdot w\,,
\end{align}
where $v$ and $w$ are the left and right boundary vectors that we are now going to fix. 
Since we defined the $[+]$ tensor by addition modulo $K$, any $v=w=\ket{z}$ with $z\in[K]$ projects on to the zero momentum transfer sector just as desired. 
To see this, assume $v=w=\ket{z}=[+]_z\ket{0}$ and in particular $v^\dagger=\bra{0}[+]_{K-z}$. Since $[+]$ is a representation of the commutative group $\mathbb Z_K$ it holds $[+]_a\cdot [+]_b=[+]_{a+b}=[+]_b\cdot [+]_a$ where again addition is modulo $K$. 
Thus we can commute $[+]_{K-z}$ through the matrix product and cancel it with its inverse $[+]_z$, i.e.,
\begin{align*}
    &v^\dagger\cdot \left[+^{(1)}\right]_{\delta_1}\cdot\left[+^{(2)}\right]_{\delta_2}\cdots \left[+^{(l)}\right]_{\delta_l} \cdot w \\
    &=\bra{0}[+]_{K-z}\cdot \left[+^{(1)}\right]_{\delta_1}\cdot\left[+^{(2)}\right]_{\delta_2}\cdots \left[+^{(l)}\right]_{\delta_l}  \cdot [+]_z\ket{0} \\
    &= \bra{0} \left[+^{(1)}\right]_{\delta_1}\cdot\left[+^{(2)}\right]_{\delta_2}\cdots \left[+^{(l)}\right]_{\delta_l} \ket{0}.
\end{align*}
Thus, we have shown that we can construct the $\delta$-MPS in terms of the $[+^{(j)}]$ tensors, which implies that we can construct the $\delta$-MPS with bond dimension 
\begin{equation}
    D_\delta= K \leq 4 k_{\mathrm{max}} n_{\mathrm{max}}+1. 
\end{equation}
\end{proof}

To conclude this section, we define $n= \max\left\{n_{\mathrm{max}}, n_{\mathrm{ZM}} \right\}$. 
Then, the complexity of each sweep of the algorithm is bounded by 
\begin{align}
    O(\chi^3k_{\mathrm{max}}^2 n^2+\chi^2k_{\mathrm{max}}^3 n^4)
\end{align}
arithmetic operations. 
Moreover, heuristically, we find that it suffices to chose $\chi = O(D) = O(k_{\mathrm{max}} n_{\mathrm{max}})$, which implies that one sweep can be performed with $O(k_{\mathrm{max}}^5 n^6)$ arithmetic operations. 
Finally, assuming $n_\mathrm{max} = k_\mathrm{max}$ with the implication that the largest mode can be at most occupied once and $n_{\mathrm{ZM}} = O(n_{\mathrm{max}})$, we obtain $O(k_\mathrm{max}^{11})$ arithmetic operations.


\section{Visualisation of quantum states}


\subsection{Wigner quasiprobability distribution}

The Wigner function is a quasiprobability distribution representing a quantum state in phase space. 
It encodes the same information as the wave function or more generally the density matrix of the state and provides an insightful visualisation allowing comparison with classical statistical ensembles which are characterised by probability distributions. 

For a system consisting of a single particle, the Wigner function of a density matrix $\rho$ is defined as 
\begin{equation}
    W(q,p)=\frac{1}{2\pi}\int_{-\infty}^{+\infty}\langle q-\tfrac{1}{2}s|\rho|q+\tfrac{1}{2}s\rangle\mathrm{e}^{\mathrm{i}ps}\,\mathrm{d}s\label{eq:Wigner_def}
\end{equation}
where $q,p$ are the generalised position and conjugate momentum, respectively. 
Integrating $(q,p)\mapsto W(q,p)$ over $q$ or $p$ gives the probability distribution of finding the particle at generalised momentum $p$ or generalised position $q$, respectively 
\begin{align*}
    \int_{-\infty}^{+\infty}W(q,p)\,\mathrm{d}q & =\langle p|\rho|p\rangle, \\
    \int_{-\infty}^{+\infty}W(q,p)\,\mathrm{d}p & =\langle q|\rho|q\rangle
\end{align*}
from which the normalisation condition 
\begin{equation}
    \int_{-\infty}^{+\infty}\int_{-\infty}^{+\infty}W(q,p)\,\mathrm{d}q\mathrm{d}p=\mathrm{Tr}(\rho)=1
\end{equation}
follows. 
Based on the above properties the Wigner function is analogous to the joint probability distribution in phase space of a classical particle in a statistical ensemble. 
Unlike the classical case it is not a strictly positive function and the presence of negative values in some phase space region is a signature of the quantum nature of the state.

The Wigner function can be expressed as the double Fourier transform of the expectation value in the state $\rho$ of the exponential operator $\mathrm{e}^{-\mathrm{i}u\hat{q} - \mathrm{i}v\hat{p}}$~\cite{Leonhardt1997}. 
Explicitly, from the above definition, it can be shown that 
\begin{equation}
    W(q,p)=\iint_{-\infty}^{+\infty}\mathrm{Tr}\left\{ \rho\mathrm{e}^{-\mathrm{i}u\hat{q}-\mathrm{i}v\hat{p}}\right\} \mathrm{e}^{\mathrm{i}uq+\mathrm{i}vp}\,\frac{\mathrm{d}u\mathrm{d}v}{(2\pi)^{2}}
    \label{eq:Wigner_computation_formula}.
\end{equation}
In a harmonic oscillator basis with frequency $\omega$ and ladder operators $a,a^{\dagger}$, the operator $\mathrm{e}^{-\mathrm{i}u\hat{q}-\mathrm{i}v\hat{p}}$ is nothing but the displacement operator 
\begin{equation}
    D(\alpha) := \mathrm{e}^{\alpha a^{\dagger}-\alpha^{*}a} 
    = \mathrm{e}^{-|\alpha|^2/2} \mathrm{e}^{\alpha a^{\dagger} } \mathrm{e}^{-\alpha^{*}a}
    \label{eq:displacement_op}
\end{equation}
with 
\begin{equation}
    \alpha=v\sqrt{\omega/2}-\mathrm{i}u/\sqrt{2\omega}\label{eq:Wigner_integr_var}.
\end{equation}

In the multi-mode case with generalised positions and momenta $\boldsymbol{q} = (q_{1},q_{2},\dots,q_{n})$ and $\boldsymbol{p} = (p_{1},p_{2},\dots,p_{n})$, respectively, the above definition generalises to
\begin{equation}
    W(\boldsymbol{q},\boldsymbol{p})=\frac{1}{2\pi}\int_{-\infty}^{+\infty}\langle\boldsymbol{q}-\tfrac{1}{2}\boldsymbol{s}|\rho|\boldsymbol{q}+\tfrac{1}{2}\boldsymbol{s}\rangle\mathrm{e}^{\mathrm{i}\boldsymbol{p}\cdot\boldsymbol{s}}\,\mathrm{d}\boldsymbol{s}.
\end{equation}
$W(\boldsymbol{q},\boldsymbol{p})$ is clearly the product of Wigner functions $W(q_{i},p_{i})$ of the reduced density matrices $\rho_{i}$ corresponding to each of the modes. 
In the present case of bosonic QFT, the Wigner function of mode $k$ corresponding to a state $\rho$ is 
\begin{align*}
    & W_{\rho}^{(k)}(\phi,\pi) \\
    & \quad = \iint_{-\infty}^{+\infty}\mathrm{Tr}\left\{ \rho_{k}\mathrm{e}^{-\mathrm{i}u\tilde{\Phi}_{k}-\mathrm{i}v\tilde{\Pi}_{k}}\right\} \mathrm{e}^{\mathrm{i}u\phi+\mathrm{i}v\pi}\,\frac{\mathrm{d}u\mathrm{d}v}{(2\pi)^{2}}
\end{align*}
and can be computed by evaluating the expectation value of the exponential operator for a suitably chosen range of values of the variables $u$ and $v$, and numerically computing the double Fourier transform.


\subsection{Full counting statistics}

The \emph{full counting statistics} (FCS) of an observable in a quantum state is the probability distribution of measurements of the observable in the given state. 
It provides complete information about the fluctuations of measured values beyond that of statistical moments like the expectation value and variance. 
The probability $P_{\rho}^{\hat{X}}(x)$ of obtaining the value $x$ in a measurement of observable $\hat{X}$ in state $\rho$ is 
\begin{equation}
    P_{\rho}^{\hat{X}}(x)=\sum_{|x'\rangle}\langle x'|\rho|x'\rangle\delta(x'-x)
\end{equation}
where the sum runs over all eigenvectors $|x'\rangle$ of the operator $\hat{X}$ with corresponding eigenvalue $x'$. 
We can formally write
\begin{equation}
    P_{\rho}^{\hat{X}}(x)=\mathrm{Tr}\left\{ \rho\hat{\delta}(\hat{X}-x)\right\} 
\end{equation}
where $\hat{\delta}(\hat{X}-x)$ is an operator valued version of the $\delta$-function. 
In practice, computing the probability distribution is easier (especially for observables whose spectrum is continuous) if we first compute the generating function, which is the Fourier transform of the probability distribution. 
It can be easily shown that 
\begin{equation}
    P_{\rho}^{\hat{X}}(x)=\int_{-\infty}^{+\infty}\mathrm{Tr}\left\{ \rho\mathrm{e}^{\mathrm{i}\hat{X}s}\right\} \mathrm{e}^{-\mathrm{i}xs}\,\frac{\mathrm{d}s}{2\pi}.
\end{equation}
For example, considering as observable the local bosonic field $\Phi$ in the mS model, the FCS in a state $\rho$ is 
\begin{equation}
    P_{\rho}^{\Phi}(\phi)=\int_{-\infty}^{+\infty}\mathrm{Tr}\left\{ \rho\mathrm{e}^{\mathrm{i}\Phi(0)s}\right\} \mathrm{e}^{-\mathrm{i}\phi s}\,\frac{\mathrm{d}s}{2\pi},
\end{equation}
and can be computed by evaluating expectation values of an imaginary exponential operator for different values of the exponent coefficient $s$ and performing the numerical Fourier transform of the resulting function. 
This computation is similar to the Wigner function with the main difference being that the exponential operator involves all modes instead of only one.

\section{Analysis of numerical data and cutoff extrapolation} 

For the plots presented in the main text, we computed MPS representations of ground and excited states and quench dynamics states for different parameter values and cutoffs $k_\mathrm{max}$. 
From these we computed eigenstate energies, expectation values of local observables, momentum space entanglement entropies, single-mode Wigner functions and FCS of local observables. 
Note that in the momentum space decomposition used in the present method, observables that are local in coordinate space correspond to operators that act on all momentum modes, therefore their computation involves contractions on all MPS sites. 
Nevertheless, for all practical purposes the operators needed are either sums or products of single mode operators over all modes. 
The entanglement entropy of a subsystem (subset of modes) with its complement is generally computed by contracting the physical indices of the density matrix corresponding to the complement to obtain the reduced density matrix of the subsystem and computing its singular values. 
When the partitioning of the system is realised by cutting a single MPS bond, then the entanglement entropy can be computed more easily as the singular values at the cut. 

For the computation of energy gaps we compute expectation values of the Hamiltonian in the MPS representations of the ground and first excited state using DMRG~\cite{Stoudenmire2012} for a sufficiently large system size (\mbox{$L=15$} for the sG model, \mbox{$L = 100$} for the mS model) such that finite size effects are negligible. 
On the other hand, corrections due to finite $k_\mathrm{max}$ are significant and in order to obtain accurate estimates of the energy gap we apply extrapolation to the limit $k_\mathrm{max} \to \infty$.

For the extrapolation of the momentum cutoff we do the following procedure: 
Once the maximum value of $k_\mathrm{max}$ that is accessible for a specific computation has been specified and $n(k)$ has been set according to the empirical rule (\ref{n(k)}), we do simulations at lower $k_\mathrm{max}$ keeping the occupation cutoffs $n(k)$ fixed. 
The results obtained in this way are then fitted to a suitable scaling function of $k_\mathrm{max}$ and extrapolated to infinite $k_\mathrm{max}$.

RG theory provides predictions for the scaling of local observables with the UV cutoff, however there are different realisations of UV truncation in RG theory and they are not identical to the one used here. Nevertheless, it is generally expected that the precise truncation method is irrelevant. 
For the energy difference between low excited states corrections are expected to decay proportionally to the inverse of the cutoff~\cite{Konik_NRG} and our numerical data are consistent with this scaling rule. 
We therefore fit the finite cutoff data to a $1/k_\mathrm{max}$ scaling function and obtain our estimates from the value of the fit at infinite $k_\mathrm{max}$. 

The estimated values of the energy gap of the sG model at the FF point and the critical point of the mS model were computed using the above method. 
For the estimation of the critical point of the mS model, in particular, a further step was necessary. 
Given that finite size effects become important close to the critical point and at the same time truncation effects worsen, we found it best to locate the critical point by extrapolating the $\Delta_\text{mS}(m)$ data from lower values of $m$ where accuracy is satisfactory. 
More specifically, we fit the data points $\Delta_\text{mS}$ for $m$ in the interval $[0.1,0.25]$ to a linear function and estimate the critical mass $m_\mathrm{c}$ as the value of $m$ at which the fitted line crosses the horizontal axis. 
Alternative polynomial fit functions and variations of the fitting interval edges do not modify the final results. 
The extrapolation over $m$ and over $k_\mathrm{max}$ can be performed in either order and the final results are equal. 
The estimated error of $m_\mathrm{c}$ corresponds to the combination of DMRG truncation errors and extrapolation fitting errors. 

Finite-size scaling theory for lattice simulations~\cite{FSS} provides precise predictions for the scaling of physical observables with the system size $L$ from which it is possible to derive information about the thermodynamic limit, including details about the model's critical behaviour (e.g., critical exponents). 
We expect that using a momentum-space formulation of finite-size scaling theory would significantly improve the accuracy of our estimates.

\end{document}